\documentclass[runningheads]{llncs}
\usepackage[T1]{fontenc}
\usepackage{graphicx}
\usepackage{hyperref}
\usepackage{color}

\usepackage{amsfonts}
\usepackage{amsmath}
\usepackage{amssymb}
\usepackage{paralist}
\usepackage[color=yellow]{todonotes}
\usepackage{ifsym}
\newcommand{\Tau}{\mathcal{T}}
\usepackage[capitalise,nameinlink]{cleveref}
\usepackage{thmtools} 
\usepackage{thm-restate}
\usepackage{subcaption}
\usepackage{enumerate}
\usepackage[inline]{enumitem}

\Crefname{algorithm}{Algorithm}{Algorithms}
\Crefname{section}{Section}{Sections}
\Crefname{lemma}{Lemma}{Lemmata}
\Crefname{figure}{Figure}{Figures}
\Crefname{property}{Property}{Properties}
\Crefname{theorem}{Theorem}{Theorems}

\newcommand{\ssquare}{{\scalebox{0.5}{$\square$}}}

\usepackage[ruled,linesnumbered]{algorithm2e}
\definecolor{mygreen}{rgb}{0, 0.5, 0}

\begin{document}
\title{Ranking and Unranking of the Planar Embeddings of a Planar Graph}
%
%
\author{Giuseppe Di Battista\inst{1} \and
Fabrizio Grosso\inst{1} \and\\
Giulia Maragno\inst{1} \and Maurizio Patrignani\inst{1}}
\authorrunning{G. Di Battista et al.}
%
\institute{University of Roma Tre \email{\{giuseppe.dibattista, fabrizio.grosso, maurizio.patrignani\}@uniroma3.it, giu.maragno@stud.uniroma3.it}}
\maketitle              
\begin{abstract}
Let $\mathcal{G}$ be the set of all the planar embeddings of a (not necessarily connected) $n$-vertex graph $G$. 
We present a bijection $\Phi$ from $\mathcal{G}$ to the natural numbers in the interval $[0 \dots |\mathcal{G}| - 1]$. 
Given a planar embedding $\mathcal{E}$ of~$G$, we show that $\Phi(\mathcal{E})$ can be decomposed into a sequence of $O(n)$ natural numbers each describing a specific feature of~$\mathcal{E}$. The function $\Phi$, which is a ranking function for $\mathcal{G}$, can be computed in $O(n)$ time, while its inverse unranking function $\Phi^{-1}$ can be computed in $O(n \alpha(n))$ time. The results of this paper can be of practical use to uniformly at random generating the planar embeddings of a graph~$G$ or to enumerating such embeddings with amortized constant delay. Also, they can be used to counting, enumerating or uniformly at random generating constrained planar embeddings of~$G$.

\end{abstract}
\section{Introduction}
The space of the planar embeddings of planar graphs is one of the favorite playgrounds of Graph Drawing and of Graph Algorithms. For instance, several algorithms (e.g., \cite{adp-fmdepg-11,bdd-codmn-00,dlop-oodlt-20,gmw-iepg-01,mw-oacep-99,mw-coepg-00}) delve into this space to find graph drawings that are optimal according to some metrics, there are problems that are polynomially solvable (e.g., \cite{DBLP:journals/algorithmica/BertolazziBLM94,dlp-hvpac-19}) if a planar embedding is given and that become NP-complete when the choice of any planar embedding is allowed (e.g., \cite{dlp-hvpac-19,gt-ccurp-01}), and there are many algorithms that look for planar drawings whose underlying planar embedding must satisfy some constraints (e.g., \cite{adfjkpr-tppeg-j14,pt-mdge-00,DBLP:journals/constraints/Tamassia98}).

This space has been studied from various perspectives. As an example, all the planar embeddings of a biconnected planar graph are implicitly represented by a data structure called SPQR-tree \cite{DBLP:journals/siamcomp/BattistaT96,DBLP:journals/algorithmica/BattistaT96}, and ranking and unranking functions where given for such a graph relying on PQ-trees or segment graphs \cite{KARABEG1993249,Vo-Dick-Williamson-85}.
Also, the results of Cai~\cite{cai1993counting} (see also \cite{DBLP:journals/dm/Stallmann93}), which are extremely important in the perspective of our paper, measure the ``extension'' of such a space, providing the exact number of planar embeddings of a given planar graph.

Let $G$ be a (not-necessarily connected) planar graph. In this paper we construct a bijective function $\Phi$ that maps a planar embedding of $G$ to a natural number (\emph{ranking} \cite{DBLP:books/KreherS99}) and such that its inverse maps a natural number to a planar embedding of $G$ (\emph{unranking} \cite{DBLP:books/KreherS99}). Both $\Phi$ and its inverse are efficiently computable.
Also, given a planar embedding $\mathcal{E}$ of~$G$, we show that $\Phi(\mathcal{E})$ can be decomposed into a sequence of $O(n)$ natural numbers each describing a specific feature of~$\mathcal{E}$. 
These results can be of practical use to uniformly at random generate the planar embeddings of~$G$ or to enumerate such embeddings with amortized constant delay. Also, they can be used for counting, enumerating or uniformly at random generating constrained planar embeddings of~$G$.

Formally, the  paper is devoted to prove \cref{th:general}. In fact, \cref{th:general} and the simple \cref{le:tuple-number} that is shown in \cref{se:preliminaries} imply the existence of a ranking and an unranking function that can be efficiently computed.

Let $G$ be a graph and denote by $n(G)$, $m(G)$, and $f(G)$, the number of vertices, edges, and faces of any planar embedding of $G$, respectively. Also, for a cut-vertex $v_\xi$ of $G$ let $\delta_{v_\xi}$ be the degree of $v_\xi$ in $G$, let $b(v_\xi)$ be the number of biconnected components containing $v_\xi$, and let $\delta_{v_\xi,\sigma}$ be the degree of $v_\xi$ restricted the biconnected component $B_\sigma$. Finally, for a parallel triconnected component $P_\xi$ of a biconnected component of $G$ let $\delta(P_\xi)$ be the number of its branches.

\begin{theorem}\label{th:general}
Let $G$ be an $n$-vertex planar graph with connected components $G_1, \dots, G_t$. 
Let $v_1, \dots, v_w$ be the cut-vertices of $G$. 
%
%
Let $P_1, \dots, P_y$ be the parallel triconnected components of the SPQR-trees of all the biconnected components of $G$. 
Let $R_1, \dots, R_z$ be the rigid triconnected components of the SPQR-trees of all the biconnected components of $G$. 
There exists a bijective function $\Phi$ whose domain is the set of the planar embeddings of $G$ on the sphere and whose codomain is the set of tuples of natural numbers 

\begin{center}
\begin{tabular}{ c }
$\langle a_1, \dots, a_{t-1}, b_1, \dots, b_t, c_{1,1}, \dots, c_{1,b(v_1)}, \dots, c_{w,1}, \dots, c_{w, b(v_w)},$ \\
$d_{1,1}, \dots, d_{1,b(v_1)-2}, \dots, d_{w,1}, \dots, d_{w,b(v_w)-2}, p_1, \dots, p_y, r_1, \dots, r_z \rangle$   
\end{tabular}
\end{center}
where
\begin{itemize}
\item $a_\xi \in [0 \dots \sum_{h=1}^{t}{(f(G_h)-1)]}$, for $\xi=1, \dots, t-1$;
\item $b_\xi \in [0 \dots f(G_h)-1]$, for $\xi=1, \dots, t$;
\item $c_{\xi,\sigma} \in [0 \dots \delta_{v_\xi,\sigma}-1]$, for $\xi = 1, \dots, w$ and $\sigma = 1, \dots, b(v_\xi)$;
\item $d_{\xi, \sigma} \in [0 \dots \delta_{v_\xi}-\sigma-1]$, for $\xi = 1, \dots, w$ and $\sigma = 1, \dots, b(v_\xi)-2$;
\item $p_\xi \in [0 \dots (\delta(P_\xi)-1)! - 1]$, for $\xi=1, \dots, y$;
\item $r_\xi \in [0, 1]$, for $\xi=1, \dots, z$.
\end{itemize}
%
The number of elements of the tuples are $O(n)$.
The function $\Phi$ can be computed in $O(n)$ time, while $\Phi^{-1}$ can be computed in $O(n \alpha(n))$ time, where $\alpha$ is the inverse Ackermann function.
\end{theorem}



The time bounds discussed in this paper assume the RAM model with uniform costs, as it is usual in the literature whenever the involved numbers can be represented with $O(\log n)$ bits~\cite{Aho1974}. This same model is sometimes assumed even when ranking permutations~\cite{DBLP:journals/ipl/MyrvoldR01}, a task that implies representing numbers up to $O(n!)$, hence using $O(\log (n!)) = O(n \log n)$ bits. Indeed, many libraries could be used to overcome this limitation with a little increase in cost per operation.

\cref{th:general} is proved in three steps.
First, in \cref{th:biconnected} we rank and unrank the planar embeddings of a biconnected component of $G$. This is simply achieved by exploiting the SPQR-tree associated with the biconnected component and by using a technique (see \cite{DBLP:journals/ipl/MyrvoldR01}) that allows to rank or unrank a permutation in linear time. More details are given in \cref{se-app:biconnected-graphs}.
Second, in \cref{se:simply-connected} we rank and unrank the arrangements of the planar embeddings of the biconnected components sharing a specific cut-vertex. To do that, we develop a new technique that associates a natural number to the position of a component around the cut-vertex.
Finally, in \cref{se:non-connected}, we rank and unrank the nesting of the connected components of $G$. To do that we devise a variation of the Pr{\"u}fer sequences, targeted to describe such a nesting. We believe this technique can have other applications. Observe that variations of the Pr{\"u}fer sequences have been already studied for different problems. For instance, in \cite{k-epcac-03} a variation is studied to represent the arrangements of the biconnected components into a block-cutvertex tree.
Note that, although related to this paper, the cited results in~\cite{cai1993counting} do not introduce efficient techniques that, given a planar embedding of $G$, can associate a natural number with the embedding and vice-versa.
We also observe that our techniques require identifying several components of $G$. The identification techniques that we adopt are partly introduced in \cref{se:preliminaries} and partly shown where they are first used.
An example that illustrates how \cref{th:general} works is shown in \cref{fig:example}

%
%
%

\section{Preliminaries}\label{se:preliminaries}

We assume familiarity with the basic concepts of graphs and planarity and, therefore, we report here only the definitions that will be used extensively in this work. For further information, the reader could refer to \cite{DBLP:books/ph/BattistaETT99,DBLP:reference/crc/2013gd}. We assume all the graphs to be undirected.

\noindent{\bf Connectivity.} The structure of a graph $G = (V,E)$ can be analyzed in terms of connectivity of its vertices. $G$ is \emph{connected} if for each $u, v \in V$ with $u \neq v$ it exists a path in $G$ connecting $u$ and $v$. A vertex whose removal makes $G$ non-connected is a \emph{cut-vertex}.
A non-connected graph can be subdivided into maximal connected subgraphs, called \emph{connected components}. In general, a graph is $k$-connected if for each pair of distinct vertices $u, v \in V$ there exist $k$ vertex-disjoint paths connecting $u$ and $v$. 2-connected and 3-connected graphs are also called \emph{biconnected} and \emph{triconnected}, respectively. It is possible to subdivide a connected (biconnected) graph into maximal biconnected (triconnected) subgraphs, called biconnected (triconnected) components. The \emph{block-cutvertex tree} of a connected graph (see \cref{sse:block-cutvertex-tree}) shows the incidence relationships between its \emph{blocks} (biconnected components) and its cut-vertices. A planar biconnected graph can be decomposed into its triconnected components. This decomposition is represented by its SPQR-tree (see \cref{sse:spqr-trees}).

\noindent{\bf Planarity and Embeddings.} Given a graph $G=(V,E)$, a \emph{drawing of $G$ on the sphere} is a mapping of each vertex in $V$ to a point of the sphere and of each edge $(u,v)$ to a Jordan curve on the sphere from the point corresponding to $u$ to the point corresponding to $v$. A drawing on the sphere is \emph{non-intersecting} (or \emph{planar}, with a slight bending of the terminology) if no two edges intersect except at common endpoints.  
A graph is \emph{planar} if it admits a planar drawing on the sphere. A drawing $\Gamma$ of a graph on the sphere determines a subdivision of the sphere into connected regions, called \emph{faces}, and a circular ordering of the edges incident to each vertex, called \emph{rotation system}. 
Visiting the (not necessarily connected) border of a face $f$ of $\Gamma$ in such a way to keep
$f$ to the left, determines a set of circular lists of vertices. Such a set is the \emph{boundary}
of $f$. Two drawings on the sphere are equivalent if they have the same rotation system and the same
face boundaries. An \emph{embedding} is an equivalence class of planar drawings on the sphere. 
Given a drawing $\Gamma$ of $G$ on the sphere and a point $p$ internal to a face $f$, the \emph{stereographic projection} \cite{stereo} of $\Gamma$ on the plane with $p$ as center of projection produces a \emph{drawing on the plane} $\Gamma'$ of $G$, where $f$ is the unbounded face, called \emph{outer face}.
Namely, let $P$ be the plane tangent to the sphere at the point diametrically opposed to $p$. Each point $x \neq p$ of the sphere is mapped to the point $x'$ of $P$ that lies on the line through $p$ and $x$.
Observe that if no two edges of $\Gamma$ cross, i.e., if $\Gamma$ is planar, also no two edges of $\Gamma'$ cross (we use the same term \emph{planar} also for $\Gamma'$).  
All drawings on the plane obtained by projecting $\Gamma$ preserve the same embedding as $\Gamma$, irrespective to the face chosen as the outer face. An \emph{embedding on the plane} is the equivalence class of the drawings on the plane that share the same embedding and the same outer face. In the following all the projections on the plane  will be stereographic.

\noindent{\bf Identifiers.} Given an $n$-vertex graph $G$ we have \emph{vertex identifiers} that uniquely associate each vertex with an integer in $[1 \dots n]$. We also have \emph{edge identifiers}: an edge $(u,v)$ is identified by the pair of identifiers of $u$ and of $v$ in ascending order. Hence, often when referring to an edge $(u,v)$ we assume that the identifier of $u$ is smaller than the identifier of $v$.
The edges of the skeletons of SPQR-trees are identified in the same way. All SPQR-trees that we consider are rooted at a reference edge, which is the edge with the minimum identifier.

\noindent{\bf Pr{\"u}fer sequences.} An example of ranking algorithm, that will be used later in this work, is given by the \emph{Pr{\"u}fer sequence}.
In \cite{prufer1918neuer} Pr{\"u}fer demonstrated a bijection between unrooted unordered trees of $n$ nodes, with nodes labeled by distinct values in $[1\dots n]$, and tuples of length $n-2$ of values in $[1\dots n]$. The result was used to demonstrate the Cayley's formula \cite{cayley1878theorem}, that states that the number of unrooted labeled trees of $n$ nodes is $n^{n-2}$. An explanation of the algorithm is given in \cref{se-app:prufer}.
%





The following lemma establishes a bijection between natural numbers and tuples of natural numbers whose elements have bounded values, generalizing to tuples a result that was already known for decision trees \cite[Theorem 3.2]{bender2005foundations}. The proof can be found in \cref{se-app:rank-tuple}.

\begin{restatable}{lemma}{leTuple}\label{le:tuple-number}
Let $P$ be the set of all tuples $\langle b_1, b_2, \dots, b_n \rangle$ such that $n \geq 1$, $0 \leq b_i < B_i$ ($i=1 \dots n$), and such that $b_i$ and $B_i$ are natural numbers. There exists a bijective function $\psi$ with domain $P$ and codomain the natural numbers in $[0 \dots \left( \prod_{i=1}^n B_i \right) -1]$. Both $\psi$ and its inverse can be computed in $O(n)$ time.
\end{restatable}

\noindent{\bf Biconnected Graphs.}\label{se:biconnected-graphs}
Let $\mu$ be a node of the SPQR-tree $\mathcal{T}^T$ of a biconnected planar graph $G$ rooted at the edge with minimum identifier. Denote by $d(\mu)$ the depth of $\mu$ in $\mathcal{T}^T$ and by $e(\mu)$ the id of the edge with the minimum identifier in the pertinent graph of $\mu$. We assign to $\mu$ the identifier $\langle d(\mu), e(\mu) \rangle$. We call \emph{conventional order} the order of the nodes of $\mathcal{T}^T$ induced by their identifiers. 

\begin{restatable}{theorem}{thBiconnected}\label{th:biconnected}
Let $G$ be an $n$-vertex biconnected planar graph and denote by $\nu_1, \dots, \nu_y$ and by $\mu_1, \dots, \mu_z$ the $P$-nodes and the $R$-nodes of the SPQR-tree~$\mathcal{T}^T$ of $G$ in conventional order, respectively. There exists a bijective function $\chi$ whose domain is the set of the planar embeddings of~$G$ and whose codomain is the set of tuples of natural numbers $\langle p_1, \dots, p_y, r_1, \dots, r_z \rangle$, where $p_\xi \in [0 \dots (\delta(\nu_\xi)-1)! - 1]$ for $\xi=1, \dots, y$ and $r_\xi \in [0\dots1]$ for $\xi=1, \dots, z$. Both the function $\chi$ and its inverse $\chi^{-1}$ can be computed in $O(n)$ time.
\end{restatable}

\section{Simply Connected Graphs}\label{se:simply-connected}

Let $G$ be a connected graph and let $B_1, B_2, \dots, B_x$ be the set of its biconnected components. Assume that each biconnected component $B_j$ has a fixed planar embedding on the sphere~${\cal E}^\circ_j$. 
We describe a bijection from the embedding of $G$ on the sphere, restricted to those that preserve the embedding ${\cal E}^\circ_j$ of each block $B_j$, and a sequence of natural numbers, where each number describes the arrangement choices around each cut-vertex of $G$ independently.
For each cut-vertex $v$, consider the $b(v)$ biconnected components of $G$ that share $v$ and rename them $B_1, B_2, \dots, B_{b(v)}$. 
Cai proved in~\cite[Lemma 13]{cai1993counting} the following formula for the number $E_v$ of the distinct planar arrangements on the sphere of the embeddings ${\cal E}^\circ_j$ around $v$, where $\delta_v$ is the degree of~$v$ and $\delta_{v,j}$ is the degree of $v$ restricted to the biconnected component $B_j$.
\begin{equation}
    E_v = \prod_{j=1}^{b(v)} \delta_{v,j} \prod_{j=1}^{b(v)-2}(\delta_v-j)
\end{equation}

We now state the main results of this section. We then provide a procedure to transform a tuple into an embedding of the graph (\cref{ssec:tuple2embedding}) and vice versa (\cref{ssec:embedding2tuple}). The correctness of the procedures is discussed in \cref{se-app:connected}. 

We have the following counterpart of Lemma 13 of~\cite{cai1993counting}.

\begin{restatable}{theorem}{thConnected}\label{th:cut-vertex}
    Let $G$ be a connected planar graph and let $v$ be a cut-vertex of $G$. Let $B_1, B_2, \dots, B_{b(v)}$ be the biconnected components of $G$ containing $v$ and let $\mathcal{E}^\circ_1, \mathcal{E}^\circ_2, \dots, \mathcal{E}^\circ_{b(v)}$ be their planar embeddings on the sphere. 
    There exists a bijection $\varphi_v$ whose domain is the set of planar embeddings on the sphere of the subgraph $G(v) = B_1 \cup B_2 \cup \dots B_{b(v)}$ preserving $\mathcal{E}^\circ_1, \mathcal{E}^\circ_2, \dots, \mathcal{E}^\circ_{b(v)}$ and whose codomain is a sequence of natural numbers $\langle c_{1}, \dots, c_{b(v)}, d_{1}, \dots, d_{b(v)-2} \rangle$, where $c_{j} \in [0 \dots \delta_{v,j}-1]$, for $j = 1, \dots, b(v)$ and $d_{j} \in [0 \dots \delta_{v}-j-1]$, for $j = 1, \dots, b(v)-2$.
    The function $\varphi_v$ can be computed in $O(n)$ time, and its inverse $\varphi_v^{-1}$ can be computed in $O(n\alpha(n))$ time, where $\alpha$ is the inverse of the Ackermann function.
\end{restatable}

Observe that the product of the ranges of the elements in the tuple of \cref{th:cut-vertex} $\prod_{i=1}^k \delta_{v,i} \prod_{h=1}^{k-2} \delta_v-h$ is exactly $E_v$ and, therefore, the bijection $\psi$ of \cref{le:tuple-number} applied to $\langle c_1, \dots, c_{b(v)}, d_1, \dots, \eta_{b(v)-2} \rangle$ produces a number in the interval $[0 \dots E_v -1]$.

Cai also proved~\cite[Theorem 4]{cai1993counting} that the number $E_G$ of distinct planar embeddings on the sphere of $G$, restricted to those that preserve the embedding ${\cal E}_j$ of each block $B_j$, is given by the product of the above numbers for all the $w$ cut-vertices of $G$, i.e., $E_G = \prod_{i=1}^{w} E_{v_i}$.
We have the following counterpart of Theorem 4 of \cite{cai1993counting} in terms of bijections.

\begin{restatable}{lemma}{leConnected}\label{le:connected-reduction-to-cut-vertex}
Let $G$ be a connected graph and let $\mathcal{E}_1, \mathcal{E}_2, \dots, \mathcal{E}_x$ be planar embeddings of its biconnected components $B_1, B_2, \dots, B_x$. 
For each cut-vertex $v_i$ of $G$, $i=1, \dots, w$, let $\varphi_{v_i}$ be a bijection from the arrangements around $v_i$ of the embeddings of those biconnected components that are incident to $v_i$ to a number in the interval $[0 \dots E_{v_i}-1]$.
There exists a bijection $\varphi$ from the embeddings on the sphere of $G$ whose restriction to $B_j$ is ${\cal E}_j$ for $j = 1, \dots, x$ to a sequence $\langle c_{1,1}, \dots, c_{1,b(v_1)}, \dots, c_{w,1}, \dots, c_{w, b(v_w)}, d_{1,1}, \dots, d_{1,b(v_1)-2}, \dots, d_{w,1}, \dots, d_{w,b(v_w)-2} \rangle$, where $c_{\xi,\sigma} \in [0 \dots \delta_{v_\xi,\sigma}-1]$, for $\xi = 1, \dots, w$ and $\sigma = 1, \dots, b(v_\xi)$; and $d_{\xi, \sigma} \in [0 \dots \delta_{v_\xi}-\sigma-1]$, for $\xi = 1, \dots, w$ and $\sigma = 1, \dots, b(v_\xi)-2$.
\end{restatable}

By \cref{le:connected-reduction-to-cut-vertex} we have that the planar embeddings on the sphere of $G$, restricted to those that preserve the embedding ${\cal E}_i$ of each block $B_i$, can be ranked by combining the ranking function $\varphi_v$ from \cref{th:cut-vertex} for the arrangement of the biconnected components around each cut-vertex $v$ with \cref{le:tuple-number}.

\subsection{Transforming a tuple into an embedding}\label{ssec:tuple2embedding}

Let $G$ be a connected planar graph and let $v$ be a cut-vertex of $G$. Let $B_1, B_2, \dots, B_{b(v)}$ be the biconnected components of $G$ containing $v$ and let $\mathcal{E}^\circ_1, \mathcal{E}^\circ_2, \dots, \mathcal{E}^\circ_{b(v)}$ be their planar embeddings on the sphere.
We describe a transformation from a sequence of natural numbers $\langle c_{1}, \dots, c_{b(v)}, d_{1}, \dots, d_{b(v)-2} \rangle$, where $c_{j} \in [0 \dots \delta_{v,j}-1]$, for $j = 1, \dots, b(v)$ and $d_{j} \in [0 \dots \delta_{v}-j-1]$, for $j = 1, \dots, b(v)-2$, to a planar embedding on the sphere of the subgraph $G(v) = B_1 \cup B_2 \cup \dots B_{b(v)}$ preserving $\mathcal{E}^\circ_1, \mathcal{E}^\circ_2, \dots, \mathcal{E}^\circ_{b(v)}$.

The first phase is to use $c_1, c_2, \dots c_{b(v)}$ to transform the embeddings on the sphere $\mathcal{E}^\circ_1, \mathcal{E}^\circ_2, \dots, \mathcal{E}^\circ_{b(v)}$ into embeddings on the plane $\mathcal{E}^\ssquare_1, \mathcal{E}^\ssquare_2, \dots, \mathcal{E}^\ssquare_{b(v)}$ of $B_1, B_2, \dots, B_{b(v)}$, respectively. 
Precisely, for each $B_j$, $j=1, \dots, b(v)$, let $e^j_0, \dots e^j_{\delta_{v, j}-1}$ be the edges of $B_j$ incident to $v$ ordered according to their identifiers. We select the $c_j$-th edge of this order as the \emph{first edge} of ${\cal E}^\circ_j$ and denote it $first_j$. In the remainder of this section, we will assume that the edges $e^j_0, \dots e^j_{\delta_{v, j}-1}$ are labeled according to a counter-clockwise visit of the adjacency list of $v$ in $\mathcal{E}^\circ_j$ starting from $first_j = e^j_0$ and ending with $e^j_{\delta_{v, j}-1}$, which we also denote $last_j$.
We project each $\mathcal{E}^\circ_j$ on the plane, obtaining $\mathcal{E}^\ssquare_j$, so that the outer face of $\mathcal{E}^\ssquare_j$ is the face of $\mathcal{E}^\circ_j$ that comes before $first_j$ and after $last_j$ in counter-clockwise order around $v$.

Now that all the biconnected components are independently embedded on the plane, we merge them with $b(v)-1$ merging operations. We initialize $b(v)$ partial embeddings $\mathcal{P}^\ssquare_i=\mathcal{E}^\ssquare_i$, for $i=1,\dots, b(v)$. The $j$-th merging operation ($j=2,\dots,b(v)$) inserts, in the way that is specified below, $\mathcal{P}^\ssquare_j$ (\emph{source} of the merge) into a suitably chosen $\mathcal{P}^\ssquare_q$ (\emph{target} of the merge). After the merge, $\mathcal{P}^\ssquare_j$ will not be used anymore and its edges are considered as edges of $\mathcal{P}^\ssquare_q$.
After $b(v)-1$ such merge operations, we remain with a single partial embedding of $G(v)$ on the plane where all $\mathcal{E}^\ssquare_j$ have been merged into. This embedding on the plane can then be regarded as an embedding on the sphere.
The first merge operation merges source $\mathcal{P}^\ssquare_2$ into target $\mathcal{P}^\ssquare_1$. $\mathcal{E}^\ssquare_1$ and $\mathcal{E}^\ssquare_2$ are merged in such a way that in a counter-clockwise visit of the adjacency list of $v$ starting from $first_1$ the edges of ${\cal E}^\ssquare_1$ are not interleaved with the edges of ${\cal E}^\ssquare_2$ and $first_2$ comes after $last_1$.
We now show how to merge the remaining $b(v)-2$ partial embeddings $\mathcal{P}^\ssquare_j$, $j=3, \dots, b(v)$, guided by the values $d_j$, $j=1, \dots, b(v)-2$. 
First, we create a sequence $\mathcal{S}$ of edges incident to $v$, that will be used through the $b(v)-2$ iterations to suitably select the edge incident to $v$ that determines the current merge operation. Sequence $\mathcal{S}$ starts with the edges of $\mathcal{E}^\ssquare_1$ in the linear ordering starting from $first_1$ and ending with $last_1$. Then we append to $\mathcal{S}$ the edges of $\mathcal{E}^\ssquare_2$ in the linear ordering starting from $first_2$ and ending with $last_2$. Then, for $j = 3,\dots, b(v)$, we append to $\mathcal{S}$ all the edges of $\mathcal{E}^\ssquare_j$, with the exception of $first_j$, ordered according to their linear order in the embedding on the plane $\mathcal{E}^\ssquare_j$. Finally, $\mathcal{S}$ is closed by the edges $first_j$ for $j = b(v), b(v)-1, \dots, 3$. 
We assume now that the edges of $\mathcal{S}$ are labeled by their position in $\mathcal{S}$, starting from $0$. As all the edges adjacent to $v$ are in $\mathcal{S}$, we have that $|\mathcal{S}| = \delta_v$. We also store a copy of $\mathcal{S}$, called $\mathcal{S}'$. 

For $j = 3 \dots b(v)$ we use $d_{j-2}$ to perform a merge operation of $\mathcal{P}^\ssquare_j$ into a suitably-chosen partial embedding. Let $e_{d_{j-2}}$ be the edge with label $d_{j-2}$ in $\mathcal{S}$. Note that, since the upper-bound of the range $[0 \dots \delta_{v}-k-1]$ of the element $d_{k}$, for $k = 1, 2, \dots, b(v)-2$, is strictly decreasing as $k$ increases, the chosen edge $e_{d_{j-2}}$ excludes the last $j-2$ elements of $\mathcal{S}$, e.g., when choosing $e_{d_2}$, the last two cells, containing $first_3$ and $first_4$, are excluded.
We distinguish two cases:
\begin{enumerate*}[label=(\arabic*)]
    \item $e_{d_{j-2}}$ belongs to a partial embedding $\mathcal{P}^\ssquare_h$ ($1 \leq h \leq b(v)$) distinct from $\mathcal{P}^\ssquare_j$.
    \item $e_{d_{j-2}}$ belongs to $\mathcal{P}^\ssquare_j$.
\end{enumerate*}

\begin{figure}[t!]
    \captionsetup[subfigure]{justification=centering}
    \centering
    \begin{subfigure}{0.53\textwidth}
    \centering
    \includegraphics[page=11, width=\textwidth]{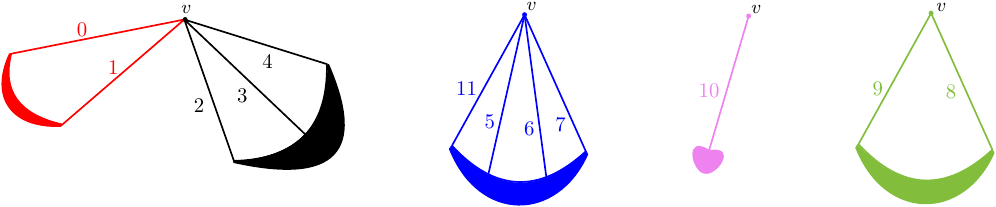}
    \subcaption{The components after the unranking of the first edges.}
    \end{subfigure}
    \hfil
    \begin{subfigure}{0.45\textwidth}
    \centering
    \includegraphics[page=12, width=\textwidth]{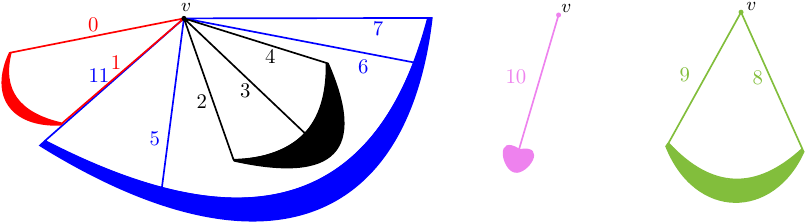}
    \subcaption{$\mu_3$ has attachment number 7, so its first edge ($e_{12}$) unifies with $e_2$ and $e_7$ goes after $e_5$.}
    \end{subfigure}
    \hfil
    \begin{subfigure}{0.4\textwidth}
    \centering
    \includegraphics[page=13, width=\textwidth]{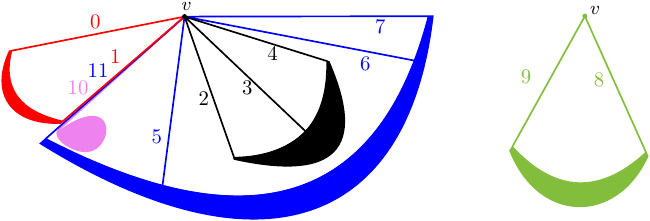}
    \subcaption{The only edge of $\mu_4$ ($e_{11}$) unifies with $e_2$ after $e_{12}$.}
    \end{subfigure}
    \hfil
    \begin{subfigure}{0.57\textwidth}
    \centering
    \includegraphics[page=14, width=0.5\textwidth]{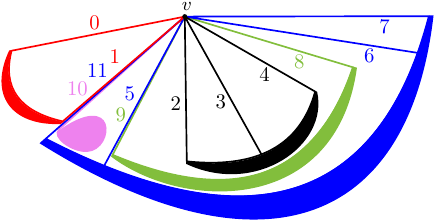}
    \subcaption{The first edge of $\mu_5$ ($e_{10}$) unifies with $e_2$ after $e_{11}$ and the edge $e_9$ goes after $e_5$, before $e_7$.}
    \end{subfigure}
    \caption{The unranking of the embedding 6-1-8.}\label{fig:unranking-6-1-8}
\end{figure}

If \emph{case 1} applies, we merge the embeddings of $\mathcal{P}^\ssquare_j$ into $\mathcal{P}^\ssquare_h$ in such a way that in a counter-clockwise visit of the adjacency list of $v$ starting from $first_h$, come all the edges of $\mathcal{P}^\ssquare_h$ up to $e_{d_{j-2}}$, then right after $e_{d_{j-2}}$ there is $first_{j}$ and consequently all the edges of $\mathcal{P}^\ssquare_j$, finally, after $last_{j}$ there are all the other edges of $\mathcal{P}^\ssquare_h$. We replace in $\mathcal{S}$ $e_{d_{j-2}}$ with $first_j$, so that $e_{d_{j-2}}$ will never be selected again and no other partial embedding will be inserted between $e_{d_{j-2}}$ and $first_j$. If $e_{d_{j-2}}=last_h$ then we update also $last_h=last_j$.

If \emph{case 2} applies, let $e_{d_{j-2}}$ be the edge in position $d_{j-2}$ in $\mathcal{S}'$. We use $\mathcal{S}'$ instead of $\mathcal{S}$ so that if $first_i$, for some $i<j$, has taken the place of $e_{d_{j-2}}$ in $\mathcal{S}$, we will not insert something between $e_{d_{j-2}}$ and $first_i$, but before $e_{d_{j-2}}$. We merge $\mathcal{P}^\ssquare_j$ with $\mathcal{P}^\ssquare_1$ by inserting $e_{d_{j-2}}$ and all the subsequent edges of $\mathcal{P}^\ssquare_j$ right after $last_2$. We will insert $first_j$ right after the edge $e^*$ that precedes $first_2$ in $\mathcal{P}^\ssquare_1$ and all following edges of $\mathcal{P}^\ssquare_j$ up to $e_{d_{j-2}}$ (not included) subsequently. Similarly to what has been done in \emph{case 1}, we then replace in $\mathcal{S}$ the edge $e^*$ with $first_j$, in such a way as to not insert any other edge between $e^*$ and $first_j$. 
Examples of the merge phase can be found in \cref{fig:unranking-0-0-0,fig:unranking-10-9-8,fig:unranking-6-1-8}.

Regarding the time taken by this algorithm, the first phase, using the $c_1, \dots, c_{b(v)}$, requires $O(n)$ time, whereas the merging phase takes $O(n\alpha(n))$ time using a union-find data structure \cite{10.1145/62.2160} to check which of the two cases applies.

\subsection{Transforming an embedding into a tuple}\label{ssec:embedding2tuple}

Let $G$ be a connected planar graph and let $v$ be a cut-vertex of $G$. Let $B_1, B_2, \dots,$ $B_{b(v)}$ be the biconnected components of $G$ containing $v$. Let $G(v) = B_1 \cup B_2 \cup \dots B_{b(v)}$. 
Let $\cal E^\circ$ be an embedding on the sphere of $G(v)$ and let $\mathcal{E}^\circ_1, \mathcal{E}^\circ_2, \dots, \mathcal{E}^\circ_{b(v)}$ be the restriction of $\cal E^\circ$ to the biconnected components $B_1, B_2, \dots, B_{b(v)}$, respectively. 
We describe a transformation from $\cal E^\circ$ to a sequence of natural numbers $\langle c_{1}, \dots, c_{b(v)}, d_{1}, \dots, d_{b(v)-2} \rangle$, where $c_{j} \in [0 \dots \delta_{v,j}-1]$, for $j = 1, \dots, b(v)$ and $d_{j} \in [0 \dots \delta_{v}-j-1]$, for $j = 1, \dots, b(v)-2$. 

First we compute the values $c_1 ,\dots, c_{b(v)}$.
For each $B_j$, $j=1, \dots, b(v)$, we consider the counter-clockwise order $\sigma_j$ of the edges incident to $v$ of $B_j$ and label the edges from $0$ to $\delta_{v,j}-1$ according to their position in $\sigma_j$ starting from the edge with minimum id. We perform a visit of the adjacency list of $v$ in counter-clockwise order starting from a random edge of ${\cal E}^\circ_1$.
The first edge of ${\cal E}^\circ_1$ encountered after all the edges of ${\cal E}^\circ_2$ will be called $first_1$ and $c_1$ will be set to its label. We perform a second visit starting from $first_1$, the first edge encountered of each ${\cal E}^\circ_j$, for $j=2,\dots,b(v)$, will be $first_j$ and $c_j$ will be set to its label.

We then label all the edges. We first put all the edges in a sequence $\mathcal{S}$ starting from $first_1$ and taking them, except for $first_i$ with $i=3, \dots, b(v)$, ordered primarily according to the increasing label of the component they belong to, and secondarily according to their counter-clockwise order around $v$ starting from the $first_1$. Then we append the edges $first_i$ with $i=3, \dots, b(v)$ to $\mathcal{S}$ in decreasing order of their component labels.
The label $\ell(e)$ of an edge $e$ is its position in $\mathcal{S}$, starting from $0$.

\begin{figure}[!tb]
    \centering
    \includegraphics[page=15, width=0.7\textwidth]{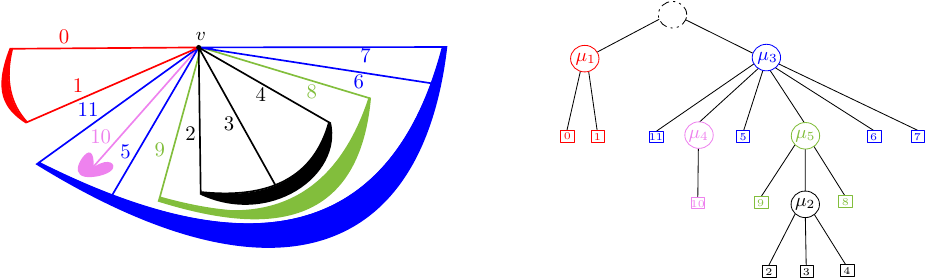}
    \caption{The planar embedding on a sphere around the cut-vertex $v$ and its resulting nesting-tree. The ranking will compute the numbers 6-1-8.}
    \label{fig:ranking-6-1-8}
\end{figure}

Next we compute an auxiliary ordered rooted tree $\Tau^B$ as follows. The nodes of $\Tau^B$ correspond either to biconnected components in $B_1, B_2, \dots, B_{b(v)}$ (\emph{component-nodes}) or to edges incident to $v$ (\emph{edge-nodes}).
We call $n(B_j)$ the node of $\Tau^B$ corresponding to $B_j$.
An edge-node can only be a leaf of $\Tau^B$.
Initially $\Tau^B$ contains only a dummy root $\rho$ and the \emph{current node} $\gamma$ of $\Tau^B$ is $\rho$.
We visit the adjacency list of $v$ counter-clockwise starting from $first_1$. 
When visiting an edge $e$ we distinguish four cases:
    \begin{enumerate*}[label=(\roman*)]
        \item If $e$ is both the first and the last edge of a component $B_j$ we add to $\Tau^B$ a component-node $n(B_j)$ as the last child of $\gamma$ and an edge-node corresponding to $e$ as child of $n(B_j)$ and set $\gamma$ as the parent of $\gamma$.
        \item If $e$ is the first, but it is not the last, edge of a component $B_j$ we add to $\Tau^B$ a component-node $n(B_j)$ as the last child of $\gamma$ and an edge-node corresponding to $e$ as the first child of $n(B_j)$. We also set node $\gamma = n(B_j)$.
        \item If $e$ is the last, but it is not the first, edge of a component $B_j$, we add to $\Tau^B$ an edge-node corresponding to $e$ as the last child of $\gamma$ and set $\gamma$ as the parent of $\gamma$.
        \item If $e$ is neither the first nor last edge of a component, we add to $\Tau^B$ an edge-node corresponding to $e$ as the last child of $\gamma$.
    \end{enumerate*}
In addition, after the computation of $\Tau^B$, let $\pi_2$ be the path between the root $\rho$ and $n(B_2)$, we save in all nodes in $\pi_2$ a Boolean value that indicates that those nodes belong to $\pi_2$. 
We then compute a pointer $jump_j$ for each node that belongs to $\pi_2$. First, we traverse the path $\pi_2$ top-down starting from the child of $\rho$ and keeping the index $i$ of the component with maximum index visited. For each node $n(B_j)$ we set $jump_j$ as the edge corresponding to the first edge-node that is a child of $n(B_j)$ on the right of $\pi_2$. Also, we set the value $max_j=i$ after having updated the value $i=j$ if $i<j$. Then we perform a post-order traversal of $\Tau^B$ starting from the child of $\rho$ that belongs to $\pi_2$ and stopping at the parent of $n(B_2)$, looking at the children of a node in reverse order (right to left) and considering only the ones that are not on the right of $\pi_2$. During the visit, if a node $B_j$ on $\pi_2$ has as left sibling a node $B_i$ with $i>j$ we set $jump_i=jump_j$, we say that $B_j$ does not belong to $\pi_2$ and $B_i$ belongs to $\pi_2$. Otherwise, if a node $B_j$ on $\pi_2$ has a parent $B_i$ (different from $\rho$) and $max_j>j$ (i.e. there exists an ancestor of $n(B_j)$ with an index greater than $j$) we set $jump_i=jump_j$ and say that $B_j$ does not belong to $\pi_2$.

Now we can exploit $\Tau^B$ to efficiently compute the values $d_1, \dots, d_{b(v)-2}$.
For $j= 1, \dots, b(v)-2$ consider node $n(B_{j+2})$ in $\Tau^B$.
Two are the cases:
Case 1: $n(B_{j+2})$ does not belong to $\pi_2$.
Let $\beta$ be the left sibling of $n(B_{j+2})$. If $\beta$ is an edge-node, then let $e$ be the edge corresponding to $\beta$, otherwise ($\beta$ is a component-node labeled $B_i$) let $e$ be $last_i$. We set $d_j=\ell(e)$. 
Case 2: $n(B_{j+2})$ belongs to $\pi_2$.
We use the pre-computed pointers by setting $d_j=\ell(jump_{j+2})$. 
Examples can be found in \cref{fig:ranking-6-1-8,fig:ranking-10-9-8,fig:ranking-0-0-0}
After setting $d_j$, we update $\mathcal{S}$ as follows. Let $\beta$ be the left sibling of $n(B_{j+2})$. If $\beta$ is an edge-node, then let $e$ be the edge corresponding to $\beta$, otherwise ($\beta$ is a component-node $B_i$) let $e$ be $last_i$. We replace $e$ with $first_{j+2}$ in $\mathcal{S}$ and set $\ell(first_{j+2}) = \ell(e)$.

This algorithm takes $O(n)$ time as the notion of which component an edge belongs to does not change during the computation, and hence a union-find data structure is not required. Also the computation of $first_i$, $\Tau^B$, and $jump_i$, for each $i$, involves a linear number of constant operations.

\section{Non-Connected Graphs}\label{se:non-connected}

Let $G$ be a planar graph with $c$ connected components ($c \geq 1$). We sort the connected components of $G$ in increasing order, based on the identifier of their vertex with the smallest identifier.
The \emph{identifier of a component} is its position in the sorted order. Hence we can call the components as $G_1, \dots, G_c$.

Let $\cal E$ be an embedding of $G$ and let ${\cal E}_1, \dots, {\cal E}_c$ be the embeddings of $G_1, \dots, G_c$, respectively, induced by $\cal E$. 
Consider the embedding ${\cal E}_i$ of $G_i$. The \emph{label of a face} $f$ of ${\cal E}_i$ is a triple $\langle i, (u,v), b \rangle$, where $(u,v)$ is the edge of $f$ with the smallest identifier and $b \in \{0, 1\}$ is $0$ if traversing $(u,v)$ from $u$ to $v$ face $f$ is to the right of $(u,v)$ and $1$ otherwise. Observe that in some cases, when traversing $(u,v)$ from $u$ to $v$, face $f$ is both to the right and to the left of $(u,v)$. In this case $f$ is the only face containing $(u,v)$ and hence the value of $b$ does not play any identification role. We denote by $F_i$ the number of faces of ${\cal E}_i$. We can sort the faces of ${\cal E}_i$ according to their increasing label and the \emph{identifier of a face} is the position of such a face in the order. Hence, we call them $f^i_1, \dots, f^i_{F_i}$.
We select as \emph{reference face} $f$ of $\cal E$ the face that contains the circular list $f^1_1$ of ${\cal E}_1$.  

We associate with ${\cal E}$ a \emph{nesting tree} $\Tau^N_{\cal E}$, which is a rooted unordered tree labeled on the edges, and a \emph{face tuple} $\langle o_1, \dots, o_c \rangle$, where $o_i\in [1 \dots F_i-1]$ as follows.
The nodes of $\Tau^N_{\cal E}$ are the components of $G$ plus a dummy node $\rho$ which is the root of $\Tau^N_{\cal E}$ and corresponds to the reference face $f$ of $\cal E$.
To build $\Tau^N_{\cal E}$ we project $\cal E$ on the plane in such a way that $f$ is the outer face. Note that this projection determines an \emph{outer face} for each ${\cal E}_i$, which is the face of ${\cal E}_i$ that would be unbounded when all other components were removed. For each $i$, we set element $o_i$ of the face tuple as the identifier of the outer face of ${\cal E}_i$. Hence, $o_i$ is a number in $[0 \dots F_i-1]$.
All the other faces of ${\cal E}_i$ are called \emph{inner faces}.
We consider the set of all the inner faces (each with its label) of ${\cal E}_1, \dots, {\cal E}_c$. This set contains $\sum_{i=1}^c (F_i-1)$ elements. We sort the set according to the increasing order of the labels. The sorting process uses the identifier of the component as the primary element, ensuring that the inner faces of a particular component are arranged consecutively in the ordering. Let $\pi(g)$ be the position of an inner face $g$ in this ordering (starting from $1$).
There is an edge in $\Tau^N_{\cal E}$ between $\rho$ and a component $G_i$ if a face $f_i$ of ${\cal E}_i$ is a circular list of $f$. The label of edge $(\rho, G_i)$ is set to zero.
$G_i$ is a parent of $G_j$ if there exists a face of $\cal E$ that contains the circular lists of both the outer face of ${\cal E}_j$ and an inner face $g$ of ${\cal E}_i$.
The label of edge $(G_i,G_j)$ is set to $\pi(g)$.
Generalizing, given a planar graph $G$ with $c$ connected components $G_1, \dots, G_c$ ($c \geq 1$) with embeddings ${\cal E}_1, \dots, {\cal E}_c$, respectively, we can define a nesting tree $\Tau^N_G$ and a face tuple $\langle o_1, \dots, o_c \rangle$ as follows.
The face tuple is any tuple such that $o_i\in [0 \dots F_i-1]$.
$\Tau^N_G$ is any $c+1$-nodes unordered tree with root $\rho$ and such that the remaining nodes are the components of $G$.

Consider the set of integers $[1 \dots \sum_{i=1}^c (F_i-1)]$. We partition such a set into consecutive intervals $I_1 = [1 \dots \sum_{i=1}^1 (F_i-1)], I_2 = [1 + \sum_{i=1}^1 (F_i-1) \dots \sum_{i=1}^2 (F_i-1)], \dots, I_c = [1 + \sum_{i=1}^{c-1} (F_i-1) \dots \sum_{i=1}^c (F_i-1)]$. 
The edges incident to $\rho$ are labeled $0$. The remaining edges of $\Tau^N_G$ are labeled in such a way that all the edges between a parent $G_h$ ($h=0,\dots, c$) and its children have a label in $I_h$.

As we have seen before, given an embedding ${\cal E}$ of $G$, it is possible to compute the corresponding nesting tree and face tuple. Also, this can be done in $O(n)$ time using standard data structures for the embeddings.

Conversely, given a set of embeddings ${\cal E}_1, \dots, {\cal E}_c$, a nesting tree, and a face tuple, built with the above requirements, it is easy to construct a corresponding embedding ${\cal E}$ of $G$. This is done in linear time by constructing an embedding on the plane as follows.
\begin{inparaenum}[(1)]
\item The numbers specified in the face tuple are used to select an outer face for the embedding of each component.
\item The components that are children of $\rho$ are embedded in the outer face.
\item The parent-child edges of the nesting tree are used to insert the embedding of the child component into an inner face of the parent component.
\item The face of the parent component in which to insert the child component is obtained from the label of the edge.
\end{inparaenum}

We summarize the above discussion with the following lemma.

\begin{lemma}\label{le:embedding-to-tree}
Let $G$ be a planar graph and let $G_1, \dots, G_c$ be the connected components of $G$ equipped with embeddings ${\cal E}_1, \dots, {\cal E}_c$. Let $F_i$ be the number of faces of ${\cal E}_i$. There exists a bijective function $\digamma$ whose domain is the set of embeddings of $G$ and whose codomain is the set of pairs composed of a face tuple and a nesting tree. Both the function $\digamma$ and its inverse can be computed in $O(n)$ time.
\end{lemma}

As we have seen so far, an embedding is fully specified by a pair consisting of a nesting tree and of a face tuple. Also, any of such pairs univocally determines an embedding. Unfortunately, this fact solves our encoding problem only partially. Namely, a face tuple is already a very simple encoding of some features of an embedding. Conversely, a nesting tree is far from being simply encoded in a set of numbers.
Hence, in the reminder of this section we concentrate on how to efficiently encode a nesting tree. Namely we will show a bijection between nesting trees and tuples of $c-1$ numbers each chosen in the interval $[0 \dots \sum_{i=1}^c (F_i-1)]$.

Rooted unordered trees with $c+1$ nodes and with a common root $\rho$ are in one-to-one correspondence with tuples of $c$ elements each in the interval $[0 \dots c]$ via the Pr{\"u}fer encoding described in \cref{se:preliminaries}. Unfortunately, this encoding does not provide information about the internal face of the parent component that contains the embedding of its child component, thereby lacking the specification of the parent-child relationship. On the other hand, we have established an order for the internal faces of the embedding of all the components that allows to retrieve, given an element in the order, which is the component it belongs to.

From the above discussion, given a nesting tree $\Tau^N_G$, we use the following variation of the Pr{\"u}fer algorithm to determine a tuple of $c-1$ numbers. The algorithm has $c-1$ iterations. At each iteration $i$ it selects the leaf $\ell$ (i.e.\ a component) of $\Tau^N_G$ with the smallest identifier, inserts into the $i$-th element of the tuple the label of the edge incident to $\ell$ and deletes $\ell$ from $\Tau^N_G$. Since the labels are in the interval $[0 \dots \sum_{i=1}^c (F_i-1)]$, each element of the tuple belongs to such an interval. In addition, the tuple can be computed in $O(n)$ time by implementing the algorithm as described in \cite{c2000efficient,CAMINITI200797,wang2009optimal}.

Conversely, given a tuple $\tau$ of $c-1$ numbers each in the interval $[0 \dots \sum_{i=1}^c (F_i-1)]$, we construct a nesting tree $\Tau^N_G$ as follows. We perform a preprocessing step computing an additional tuple $\tau'$ with $c-1$ elements as follows. We scan $\tau$ and for $i=1,\dots, c-1$ we extract the $i$-th element $\tau_i$ of $\tau$.
If $\tau_i=0$, we insert $0$ in the $i$-th position of $\tau'$. Otherwise, suppose that $\tau_i \in I_h$.
We insert $h$ in the $i$-th position of $\tau'$. Furthermore, we associate with each component $G_i$ of $G$ the number $\delta_i$ of occurrences of $i$ in $\tau'$ plus one. Intuitively, in this way we precompute the number of nodes that will be adjacent to $G_i$ in $\Tau^N_G$. Also, we associate $\rho$ with the number $\delta_0$ of occurrences of $0$ in $\tau$ plus two. Finally, we insert into $\Tau^N_G$ a node for each component $G_i$ plus a node $\rho$.
After the preprocessing step, the algorithm proceeds with $c-1$ iterations. In each iteration $i$, we select the component $G_h$ with $ \delta_h = 1$ and the smallest $h$. We insert an edge in $\Tau^N_G$ from child $G_h$ to the parent component $G_k$ such that $k$ is the $i$-th element of $\tau'$, if $k=0$ the parent will be $\rho$. We label such an edge with the $i$-th element of $\tau$. We finally decrement $\delta_h$ and $\delta_k$. After the last iteration, there remains one component $G_l$ such that $\delta_l = 1$, while $\delta_0 = 2$. We insert an edge in $\Tau^N_G$ from child $G_l$ to the parent component $\rho$ with label $0$.

Observe that appending a $0$ to $\tau'$ results exactly in the Pr{\"u}fer encoding of a rooted tree with $c+1$ nodes having the sought adjacencies between components and $\rho$ as its root. This tree is exactly $\Tau^N_G$ without the edge labels.
Having a set of labels for the edges that is dependent only on the parent node allows us to specify a label instead of a parent node in the encoding. We summarize the above discussion in the following lemma.
    
\begin{lemma}\label{le:tree-to-code}
    Let $G$ be a planar graph and let $G_1, \dots, G_c$ be the connected components of $G$ equipped with embeddings ${\cal E}_1, \dots, {\cal E}_c$. Let $F_i$ be the number of faces of ${\cal E}_i$.
    There exists a bijective function $\psi$ whose domain is the set of nesting trees of $G$ and whose codomain are the tuples of length $c-1$ whose elements are values in the interval $[0 \dots \sum_{i=1}^c (F_i-1)]$. Both the function $\psi$ and its inverse can be computed in $O(n)$ time.
\end{lemma}
Now we are ready to prove the following theorem.

\begin{theorem}\label{th:non-connected}
Let $G$ be a planar graph with connected components $G_1, \dots, G_c$ ($c>1$). Let ${\cal E}_1, \dots, {\cal E}_c$ be planar embeddings for $G_1, \dots, G_c$ such that embedding ${\cal E}_i$ has $F_i$ faces.
There exists a bijective function $\varphi$ whose domain is the set of planar embeddings of $G$ such that each $G_i$ has embedding ${\cal E}_i$ and whose codomain is the set of tuples of natural numbers $\langle a_1, \dots, a_{c-1}, b_1, \dots, b_c \rangle$ where $a_i \in \{0, \dots,  \sum_{j=1}^{c}{(F_j-1)}\}$ and $b_i \in \{0, \dots, F_i-1\}$.
Both the function $\varphi$ and its inverse can be computed in $O(n)$ time.
\end{theorem}
\begin{proof}
    The function $\varphi$ is the combination of the function $\digamma$ from \cref{le:embedding-to-tree} and the function $\psi$ from \cref{le:tree-to-code}.
\end{proof}

\section{Conclusions and Future Work}\label{se:conclusions}

In this paper, we addressed the problem of ranking and unranking all planar embeddings of an $n$-vertex graph $G$, which may not be connected. In particular, we produced a ranking function $\Phi$ that can be computed in $O(n)$ time, and its inverse unranking function $\Phi^{-1}$ that can be computed in $O(n\alpha(n))$ time.
In addition, we showed that the natural number associated with a planar embedding can be decomposed into a sequence of a linear number of values, each associated with a specific feature of the embedding. This property has practical implications, allowing us to generate embeddings uniformly at random by independently generating each value in the sequence. It also facilitates the counting, enumeration, and generation uniformly at random of constrained embeddings.

As future work we point out the problem of devising a linear time algorithm for the unranking of simply connected graphs, whose current $O(n \alpha(n))$-time bound slightly impacts the unranking of the whole graph.
Further, we would like to extend the results of this paper to beyond planar embeddings of graphs. 
Analogously to the topological morphing between two embeddings of a biconnected graph~\cite{DBLP:journals/tcs/AngeliniCBP13}, the concepts introduced in this paper could lead to topological morphing algorithms for non-connected and simply connected graphs.


\bibliographystyle{splncs04}
\bibliography{final}

\begin{thebibliography}{10}
\providecommand{\url}[1]{\texttt{#1}}
\providecommand{\urlprefix}{URL }
\providecommand{\doi}[1]{https://doi.org/#1}

\bibitem{Aho1974}
Aho, A.V., Hopcroft, J.E., Ullman, J.D.: {The Design and Analysis of Computer
  Algorithms}. Addison-Wesley, Reading, Mass. (1974)

\bibitem{adp-fmdepg-11}
Angelini, P., {Di Battista}, G., Patrignani, M.: Finding a minimum-depth
  embedding of a planar graph in $o(n^4)$ time. Algorithmica  \textbf{60}(4),
  890--937 (2011)

\bibitem{DBLP:journals/tcs/AngeliniCBP13}
Angelini, P., Cortese, P.F., {Di Battista}, G., Patrignani, M.: Topological
  morphing of planar graphs. Theor. Comput. Sci.  \textbf{514},  2--20 (2013)

\bibitem{adfjkpr-tppeg-j14}
Angelini, P., {Di Battista}, G., Frati, F., Jel{\`{\i}}nek, V.,
  Kratochv{\`{\i}}l, J., Patrignani, M., Rutter, I.: Testing planarity of
  partially embedded graphs. ACM Transactions on Algorithms  \textbf{11}(4)
  (2015), article No. 32

\bibitem{bender2005foundations}
Bender, E.A., Williamson, S.G.: Foundations of combinatorics with applications.
  Courier Corporation (2005)

\bibitem{bdd-codmn-00}
Bertolazzi, P., {Di Battista}, G., Didimo, W.: Computing orthogonal drawings
  with the minimum number of bends. IEEE Trans. on Computers  \textbf{49},
  826--840 (2000)

\bibitem{DBLP:journals/algorithmica/BertolazziBLM94}
Bertolazzi, P., {Di Battista}, G., Liotta, G., Mannino, C.: Upward drawings of
  triconnected digraphs. Algorithmica  \textbf{12}(6),  476--497 (1994)

\bibitem{c2000efficient}
C.~Chen, H., L.~Wang, Y.: An efficient algorithm for generating pr{\"u}fer
  codes from labelled trees. Theory of Computing Systems  \textbf{33},  97--105
  (2000)

\bibitem{cai1993counting}
Cai, J.: Counting embeddings of planar graphs using dfs trees. SIAM Journal on
  Discrete Mathematics  \textbf{6}(3),  335--352 (1993)

\bibitem{CAMINITI200797}
Caminiti, S., Finocchi, I., Petreschi, R.: On coding labeled trees. Theoretical
  Computer Science  \textbf{382}(2),  97--108 (2007), latin American
  Theoretical Informatics

\bibitem{cayley1878theorem}
Cayley, A.: A theorem on trees. Quart. J. Math.  \textbf{23},  376--378 (1878)

\bibitem{DBLP:books/ph/BattistaETT99}
{Di Battista}, G., Eades, P., Tamassia, R., Tollis, I.G.: Graph Drawing:
  Algorithms for the Visualization of Graphs. Prentice-Hall (1999)

\bibitem{DBLP:journals/algorithmica/BattistaT96}
{Di Battista}, G., Tamassia, R.: On-line maintenance of triconnected components
  with {SPQR}-trees. Algorithmica  \textbf{15}(4),  302--318 (1996)

\bibitem{DBLP:journals/siamcomp/BattistaT96}
{Di Battista}, G., Tamassia, R.: On-line planarity testing. {SIAM} J. Comput.
  \textbf{25}(5),  956--997 (1996)

\bibitem{dlop-oodlt-20}
Didimo, W., Liotta, G., Ortali, G., Patrignani, M.: Optimal orthogonal drawings
  of planar 3-graphs in linear time. In: Chawla, S. (ed.) Proc. ACM-SIAM
  Symposium on Discrete Algorithms (SODA '20). pp. 806--825. ACM-SIAM (2020)

\bibitem{dlp-hvpac-19}
Didimo, W., Liotta, G., Patrignani, M.: Hv-planarity: Algorithms and
  complexity. Journal of Computer and System Sciences  \textbf{99},  72--90
  (2019)

\bibitem{gt-ccurp-01}
Garg, A., Tamassia, R.: On the computational complexity of upward and
  rectilinear planarity testing. SIAM J. Comput.  \textbf{31}(2),  601--625
  (2001)

\bibitem{gmw-iepg-01}
Gutwenger, C., Mutzel, P., Weiskircher, R.: Inserting an edge into a planar
  graph. In: Proceedings of the twelfth annual ACM-SIAM symposium on Discrete
  algorithms. pp. 246--255. SODA '01, Society for Industrial and Applied
  Mathematics, Philadelphia, PA, USA (2001)

\bibitem{DBLP:conf/gd/GutwengerM00}
Gutwenger, C., Mutzel, P.: A linear time implementation of {SPQR}-trees. In:
  Marks, J. (ed.) Graph Drawing, 8th International Symposium, {GD} 2000,
  Colonial Williamsburg, VA, USA, September 20-23, 2000, Proceedings. LNCS,
  vol.~1984, pp. 77--90. Springer (2000)

\bibitem{DBLP:journals/cacm/HopcroftT73}
Hopcroft, J.E., Tarjan, R.E.: Efficient algorithms for graph manipulation {[H]}
  (algorithm 447). Commun. {ACM}  \textbf{16}(6),  372--378 (1973)

\bibitem{k-epcac-03}
Kajimoto, H.: An extension of the {P}r{\"u}fer code and assembly of connected
  graphs from their blocks. Graphs and Combinatorics  \textbf{19},  231--239
  (2003)

\bibitem{KARABEG1993249}
Karabeg, A.: Ranking planar embeddings using pq-trees. In: Gimbel, J., Kennedy,
  J.W., Quintas, L.V. (eds.) Quo Vadis, Graph Theory?, Annals of Discrete
  Mathematics, vol.~55, pp. 249--260. Elsevier (1993)

\bibitem{Vo-Dick-Williamson-85}
Kiem-Phong~Vo, W.E.D., Williamson, S.G.: Ranking and unranking planar
  embeddings. Linear and Multilinear Algebra  \textbf{18}(1),  35--65 (1985)

\bibitem{DBLP:books/KreherS99}
Kreher, D.L., Stinson, D.R.: Combinatorial algorithms: generation, enumeration,
  and search. CRC Press (1999)

\bibitem{stereo}
Leyshon, P.R., Lisle, R.J.: Stereographic Projection Techniques. Elsevier
  Science (1996), original from the University of California

\bibitem{mw-oacep-99}
Mutzel, P., Weiskircher, R.: Optimizing over all combinatorial embeddings of a
  planar graph. In: 7th International IPCO Conference on Integer Programming
  and Combinatorial Optimization. pp. 361--376. Springer-Verlag, London, UK
  (1999)

\bibitem{mw-coepg-00}
Mutzel, P., Weiskircher, R.: Computing optimal embeddings for planar graphs.
  In: Proceedings of the 6th Annual International Conference on Computing and
  Combinatorics. pp. 95--104. COCOON '00, Springer-Verlag, London, UK (2000)

\bibitem{DBLP:journals/ipl/MyrvoldR01}
Myrvold, W.J., Ruskey, F.: Ranking and unranking permutations in linear time.
  Inf. Process. Lett.  \textbf{79}(6),  281--284 (2001)

\bibitem{pt-mdge-00}
Pizzonia, M., Tamassia, R.: Minimum depth graph embedding. In: Proc. ESA 2000.
  LNCS, vol.~1879, pp. 356--367. Springer (2000)

\bibitem{prufer1918neuer}
Prufer, H.: Neuer bewis eines satzes uber permutationnen. Arch. Math. Phys.
  \textbf{27},  742--744 (1918)

\bibitem{DBLP:journals/dm/Stallmann93}
Stallmann, M.F.M.: On counting planar embeddings. Discret. Math.
  \textbf{122}(1-3),  385--392 (1993)

\bibitem{DBLP:journals/constraints/Tamassia98}
Tamassia, R.: Constraints in graph drawing algorithms. Constraints An Int. J.
  \textbf{3}(1),  87--120 (1998)

\bibitem{DBLP:reference/crc/2013gd}
Tamassia, R. (ed.): Handbook on Graph Drawing and Visualization. Chapman and
  Hall/CRC (2013)

\bibitem{10.1145/62.2160}
Tarjan, R.E., van Leeuwen, J.: Worst-case analysis of set union algorithms. J.
  ACM  \textbf{31}(2),  245–281 (mar 1984)

\bibitem{wang2009optimal}
Wang, X., Wang, L., Wu, Y.: An optimal algorithm for prufer codes. J. Softw.
  Eng. Appl.  \textbf{2}(2),  111--115 (2009)

\end{thebibliography}

\clearpage
\appendix

\section{A full example of ranking and unranking}

\begin{figure}
    \centering
    \includegraphics[page=18, width=\textwidth]{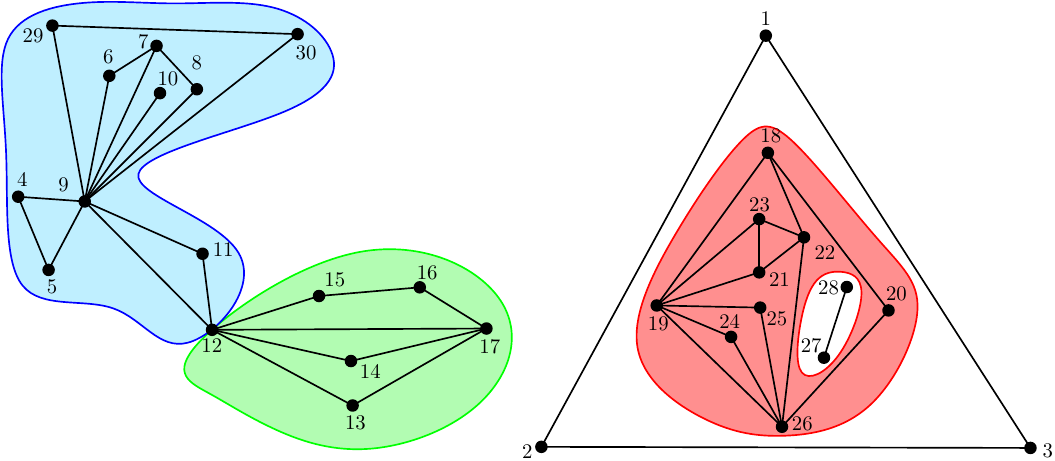}
    \caption{A graph drawn according to its $754.705.812.645^{th}$ planar embedding on the sphere, corresponding to the sequence $\langle 0,11,1,0,1,0,0,0,1,0,1,1,0,1, 2,1,6,1,4,5,0,1 \rangle$, as defined in \cref{th:general}. In particular, the bound for each number, as described in \cref{le:tuple-number}, is $\langle 17,17,17,2,9,8,1,2,3,1,2,2,2,4,9,8,7,2,6,6,2,2 \rangle$ and the number of embeddings is 19.716.667.342.848. Also, the rank restricted to the red (resp. green) biconnected subgraph is 21 (resp. 4), and the rank of the embedding around the cut-vertex 9 (in blue) is 3653.}
    \label{fig:example}
\end{figure}

\section{Graph Decomposition}

\subsection{Block-cutvertex Tree}\label{sse:block-cutvertex-tree}

Let $G$ be a connected graph and let $B_1, B_2, \dots, B_x$ be its biconnected components or blocks. The popular data structure known as the \emph{block-cutvertex tree} $\mathcal{T}^B$ of $G$, describes the decomposition of $G$ into its blocks~\cite{DBLP:journals/cacm/HopcroftT73}. Namely, $\mathcal{T}^B$ has a node $v_{B_i}$ for each block $B_i$ of $G$ and a node $v_{c_j}$ for each cutvertex $c_j$ of $G$: $v_{B_i}$ is called a \emph{block-node} of $\mathcal{T}^B$ and $v_{c_j}$ is a \emph{cutvertex-node} of $\mathcal{T}^B$. There is an arc $(v_{B_i},v_{c_j})$ in $\mathcal{T}^B$ if $c_j$ belongs to $B_i$ in $G$. 

Let $G$ be a connected graph with $n$ vertices and $m$ edges. Since each edge of $G$ uniquely belongs to a biconnected component of $G$, the overall size of the biconnected components of $G$ is $O(n+m)$, which is $O(n)$ for a planar graph. Further, since the biconnected components of $G$ can be computed in $O(n+m)$ time~\cite{DBLP:journals/cacm/HopcroftT73}, the block-cutvertex tree $\mathcal{T}^B$ of $G$ can be constructed in $O(n+m)$ time, which is $O(n)$ if $G$ is planar. 

\subsection{SPQR-tree}\label{sse:spqr-trees}

The SPQR-tree data structure, introduced by Di Battista and
Tamassia~\cite{DBLP:journals/siamcomp/BattistaT96,DBLP:journals/algorithmica/BattistaT96}, can be used to implicitly represent all planar embeddings of a biconnected planar graph $G$.  The SPQR-tree  $\mathcal{T}^T$ of $G$ represents a decomposition of $G$ into triconnected components along its split pairs.  
Each node $\mu$ of $\mathcal{T}^T$ is associated with a graph, called \emph{skeleton of $\mu$}, and denoted $skel(\mu)$. The edges of $skel(\mu)$ are either edges of $G$, which we call \emph{real edges}, or newly introduced edges, called \emph{virtual edges}.
The tree $\mathcal{T}^T$ is initialized to a single node $\mu$, whose skeleton, composed only of real edges, is~$G$. 
Consider a split pair $\{u, v\}$ of the skeleton of some node $\mu$ of $\mathcal{T}^T$,
and let $G_1,\dots,G_k$ be the components of $G$ with respect to $\{u,v\}$ such that $G_1$ is not a virtual edge and, if $k=2$, also $G_2$ is not a virtual edge.  
We introduce a node $\nu$ adjacent to $\mu$ whose skeleton is the graph $G_1 + e_{\nu,\mu}$, where $e_{\nu,\mu} = (u,v)$ is a virtual edge, and replace $skel(\mu)$ with the graph $\bigcup_{i \neq 1} G_i + e_{\mu,\nu}$, where $e_{\mu,\nu} = (u,v)$ is a virtual edge.  We say that $e_{\nu,\mu}$ is the \emph{twin virtual edge of} $e_{\mu,\nu}$, and vice versa.
Iteratively applying the above described replacement produces a tree with more nodes but smaller skeletons associated with the nodes. Eventually, when no further replacement is possible, the skeletons of the nodes of $\mathcal{T}^T$ are of four types: parallels of at least three virtual edges ($P$-nodes), parallels of exactly one virtual edge and one real edge ($Q$-nodes), cycles of exactly three virtual edges ($S$-nodes), and triconnected planar graphs ($R$-nodes). The \emph{merge} of two adjacent nodes $\mu$ and $\nu$ in $\mathcal{T}^T$, replaces $\mu$ and $\nu$ in $\mathcal{T}^T$ with a new node $\tau$ that is adjacent to all the neighbors of $\mu$ and $\nu$, and whose skeleton is $skel(\mu) \cup skel(\nu) \setminus \{ e_{\mu,\nu}, e_{\nu,\mu}\})$, where the end-vertices of $e_{\mu,\nu}$ and $e_{\nu,\mu}$ that correspond to the same vertices of $G$ are identified. By iteratively merging adjacent $S$-nodes, we eventually obtain the (unique) SPQR-tree data structure introduced by Di Battista and Tamassia~\cite{DBLP:journals/algorithmica/BattistaT96,DBLP:journals/siamcomp/BattistaT96}, where the skeleton of an $S$-node is a cycle. The crucial property of this decomposition is that a planar embedding of $G$ uniquely induces a planar embedding of the skeletons of its nodes and that, arbitrary and independently, choosing planar embeddings for all the skeletons uniquely determines an embedding of $G$. Observe that the skeletons of $S$- and $Q$-nodes have a unique planar embedding, that the skeleton of $R$-nodes have two planar embeddings (which are one the mirror of the other), and that $P$-nodes have as many planar embedding as the circular permutations of their virtual edges.
Consider a node $\mu$ and a virtual edge $e_{\nu,\mu}$ in $skel(\mu)$. Let $T_{\nu,\mu}$ be the subtree of $T$ obtained by removing the arc $(\nu,\mu)$ from $T$ and contains $\nu$.
Suppose that $\Tau^T$ is rooted at a node $\rho$. The \emph{pertinent graph} of a node $\nu$, with parent $\mu$, is the subgraph of $G$ defined as follows. First, a node $\nu'$ is obtained by iteratively merging all the nodes in the subtree of $\Tau^T$ rooted at $\nu$. Then the pertinent graph of $\nu$  corresponds to $skel(\nu')$ after the removal of the virtual~edge~$e_{\mu,\nu}$.

If $G$ has $n$ vertices, then $\mathcal{T}^T$ has $O(n)$ nodes and the total number of virtual edges in the skeletons of the nodes of $\mathcal{T}^T$ is in $O(n)$.  Also, $\mathcal{T}^T$ can be constructed in $O(n)$ time~\cite{DBLP:conf/gd/GutwengerM00}.

\section{ Pr{\"u}fer sequence}\label{se-app:prufer}
Here we describe the algorithms by Pr{\"u}fer to rank and unrank an unordered unrooted labeled tree\cite{prufer1918neuer}. Consider an alphabet $A$ made of integers in $[1 \dots n]$. 
Let $T$ be an $n$-nodes tree with nodes labeled by distinct elements of $A$ and let $\tau$ be a tuple initialized to $\emptyset$. The ranking algorithm works by iteratively deleting the leaf of $T$ with the smallest label and by appending the label of the adjacent node to $\tau$, until there are only two nodes left. For the unranking algorithm, a preliminary phase is to compute the degree of each node as the number of times its label appears in the tuple $\tau$ augmented by 1. Then for $1 \leq i \leq n-2$ find the node $v$ that is a leaf (it has degree 1) and has the minimum label and the node $u$ with label $\tau[i]$. Add an edge between $v$ and $u$ and decrease the degree of both. After all elements of the sequence have been processed, there are still two vertices with degree 1 (the others will have degree 0), simply adding an edge between those two nodes results in a unique tree. In \cite{CAMINITI200797} the authors show how adding another label at the end of the Pr{\"u}fer sequence results in rooting the tree at the specified node. In \cite{c2000efficient,CAMINITI200797,wang2009optimal} it is shown that the ranking and unranking of rooted labeled trees via the Pr{\"u}fer sequences can be done in linear time. See \cref{fig:tree-to-prufer} and \cref{fig:prufer-to-tree} for an example of ranking and unranking of a rooted tree.

\begin{figure}[tb!]
    \captionsetup[subfigure]{justification=centering}
    \centering
    \begin{subfigure}{0.15\textwidth}
    \centering
    \includegraphics[page=1, width=\textwidth]{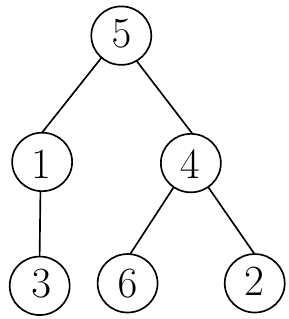}
    \subcaption{$\varnothing$}
    \end{subfigure}
    \hfil
    \begin{subfigure}{0.15\textwidth}
    \centering
    \includegraphics[page=2, width=\textwidth]{figures/prufer-example.pdf}
    \subcaption{$4$}
    \end{subfigure}
    \hfil
    \begin{subfigure}{0.15\textwidth}
    \centering
    \includegraphics[page=3, width=\textwidth]{figures/prufer-example.pdf}
    \subcaption{$4, 1$}
    \end{subfigure}
    \hfil
    \vspace{1pc}
    \begin{subfigure}{0.15\textwidth}
    \centering
    \includegraphics[page=4, width=\textwidth]{figures/prufer-example.pdf}
    \subcaption{$4, 1, 5$}
    \end{subfigure}
    \hfil
    \begin{subfigure}{0.15\textwidth}
    \centering
    \includegraphics[page=5, width=\textwidth]{figures/prufer-example.pdf}
    \subcaption{$4, 1, 5, 4$}
    \end{subfigure}
    \hfil
    \begin{subfigure}{0.15\textwidth}
    \centering
    \includegraphics[page=6, width=\textwidth]{figures/prufer-example.pdf}
    \subcaption{$4, 1, 5, 4, 5$}
    \end{subfigure}
    \caption{The ranking of a rooted labeled tree using the Pr{\"u}fer sequence.}\label{fig:tree-to-prufer}
\end{figure}

\begin{figure}[tb!]
    \captionsetup[subfigure]{justification=centering}
    \centering
    \begin{subfigure}{0.25\textwidth}
    \centering
    \includegraphics[page=7, width=\textwidth]{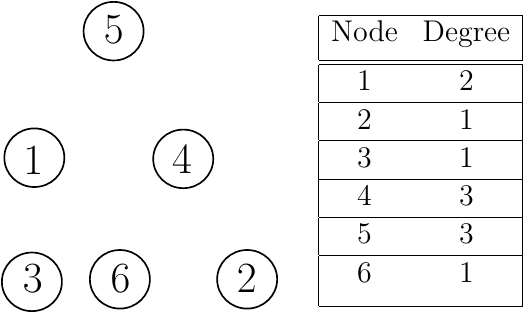}
    \end{subfigure}
    \hfil
    \begin{subfigure}{0.25\textwidth}
    \centering
    \includegraphics[page=8, width=\textwidth]{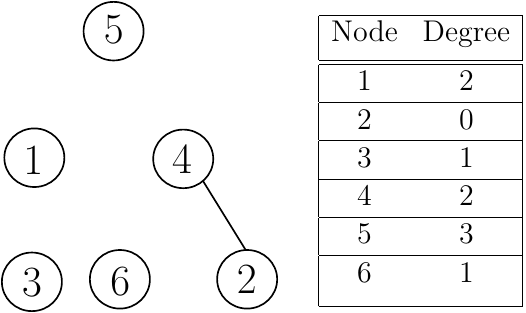}
    \end{subfigure}
    \hfil
    \begin{subfigure}{0.25\textwidth}
    \centering
    \includegraphics[page=9, width=\textwidth]{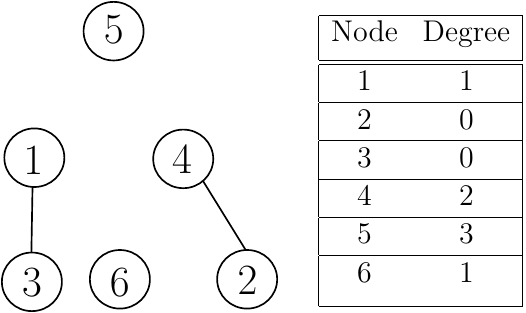}
    \end{subfigure}
    \hfil
    \vspace{1pc}
    \begin{subfigure}{0.25\textwidth}
    \centering
    \includegraphics[page=10, width=\textwidth]{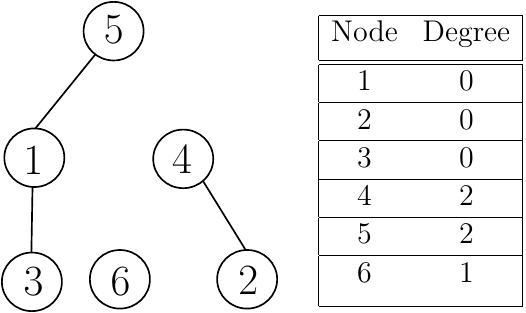}
    \end{subfigure}
    \hfil
    \begin{subfigure}{0.25\textwidth}
    \centering
    \includegraphics[page=11, width=\textwidth]{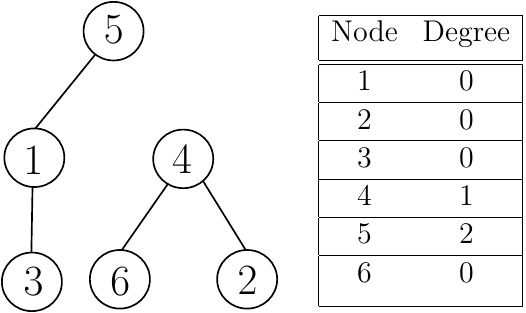}
    \end{subfigure}
    \hfil
    \begin{subfigure}{0.25\textwidth}
    \centering
    \includegraphics[page=12, width=\textwidth]{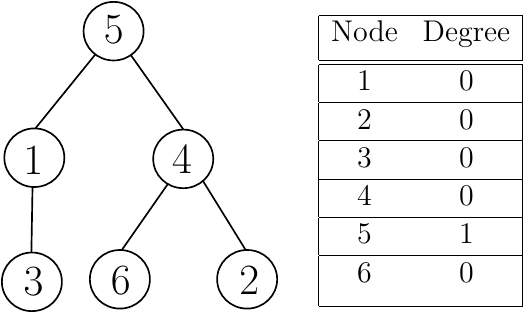}
    \end{subfigure}
    \caption{The unranking of a rooted labeled tree using the Pr{\"u}fer sequence $4,1,5,4,5$.} \label{fig:prufer-to-tree}
\end{figure}

\section{How to rank a tuple}\label{se-app:rank-tuple}
Here we give the proof of \cref{le:tuple-number}, stated in \cref{se:preliminaries}.

\leTuple*

\begin{proof}
Let $\tau$ be a tuple of length $n$ as defined in the statement of this lemma. Let $\tau_i$ denote the tuple composed by the first $i$ elements of $\tau$ ($i=1,\dots,n$).
We define the function $\psi$ recursively as follows. If $i=1$, we set $\psi(\tau_1)=\psi(\langle b_1 \rangle) = p_1 = b_1$. For $i>1$, we define $\psi(\tau_i) = p_i = \psi(\tau_{i-1}) \cdot B_i + b_i = p_{i-1} \cdot B_i + b_i$.

The function $\psi$ is surjective with respect to the codomain $[0 \dots \left( \prod_{i=1}^n B_i \right) -1]$. 
This property holds for $n=1$, where, by definition, the function can generate all the values in $[0 \dots B_{n}-1] \equiv [0 \dots \left( \prod_{i=1}^1 B_i \right) -1]$.
Let's assume that the property holds for $n \geq 1$. We will now prove that it also holds for $n+1$. Since $\psi$ is surjective for $n \geq 1$, we know that $p_n$ can assume all the values in $[0 \dots \left( \prod_{i=1}^n B_i \right) -1]$. By multiplying these values by $B_{i+1}$ and adding all possible values of $b_{i+1}$, we cover all the values in $[0 \dots \left( \prod_{i=1}^{n+1} B_i \right) -1]$.
The function $\psi$ is also injective. In fact, it is surjective and its domain and codomain have the same cardinality.

The inverse $\psi^{-1}$ of $\psi$ turns out to be as follows. If $i=1$, we have $\psi^{-1}(p_1)= \langle b_1 \rangle$. For $i>1$, we can calculate $b_i=p_i \bmod B_{i}$ and $p_{i-1}=\lfloor \frac{p_i}{B_{i}} \rfloor$. To obtain the tuple $\tau_i$, we concatenate $\psi^{-1}(p_{i-1})$ and $b_i$.
\end{proof}


\section{Biconnected Graphs}\label{se-app:biconnected-graphs}


Let $G$ be an $n$-vertex biconnected planar graph and let $\mathcal{T}^T$ be its SPQR-tree. Let $\mu$ be a node of $\mathcal{T}^T$ and consider its skeleton $skel(\mu)$ with poles $u$ and $v$. Let $d(\mu)$ be the distance of $\mu$ from the root of $\mathcal{T}^T$. Let $e(\mu)$ be the edge with the minimum identifier corresponding to a $Q$-node among those in the sub-tree of $\mathcal{T}^T$ rooted at $\mu$. Observe that if $\mu$ is a $Q$-node, then the sub-tree coincides with $\mu$, and that two different nodes $\mu_1$ and $\mu_2$ of $\mathcal{T}^T$ have different pairs $\langle d(\mu_1), e(\mu_1) \rangle$ and $\langle d(\mu_2), e(\mu_2) \rangle$. Hence, we assign the pair $\langle d(\mu), e(\mu) \rangle$ to $\mu$ as its \emph{identifier}.

If $\mu$ is an $R$-node, the \emph{first embedding} of $skel(\mu)$ is defined as follows. Let $u$ be the pole of $skel(\mu)$ with minimum identifier. Let $(u,w_1)$ be the edge incident to $u$ with minimum identifier and let $(u,w_2)$ be the edge adjacent to $(u,w_1)$ in the circular list of the edges incident to $u$ in the planar embedding of $skel(\mu)$ with minimum identifier. The first embedding of $skel(\mu)$ is such that $(u,w_1)$ precedes $(u,w_2)$ in the clockwise order around $u$. Since the skeleton of an R-node has exactly two planar embeddings, we call \emph{second embedding} the one that is not the first.
We order the $R$-nodes, and hence their skeletons, of $\mathcal{T}^T$ according to their increasing identifiers.

If $\mu$ is a $P$-node the planar embeddings of its skeleton are given by all the permutations of its virtual edges. The \emph{first embedding} of $skel(\mu)$ is defined as follows. We order the edges of $skel(\mu)$ different from its reference edge according to the identifiers of their corresponding nodes in $\mathcal{T}^T$. The clockwise order of the first embedding encounters first the reference edge and then all the other edges of the skeleton with the above order.
We order the $P$-nodes, and hence their skeletons, of $\mathcal{T}^T$ according to their increasing identifiers. 


\thBiconnected*

\begin{proof}
In Lemma 4.2 of \cite{DBLP:journals/siamcomp/BattistaT96} it is shown that the planar embeddings of $G$ are in one-to-one correspondence with all the possible embeddings of the skeletons of the R-nodes and of the P-nodes of~$\mathcal{T}^T$.
Consider the elements $r_1, \dots, r_z$ of the tuple in the statement. We assign value $r_\xi=0$ ($\xi=1,\dots,z$) if the embedding selected for $skel(\mu_\xi)$ is its first embedding and we assign value $r_\xi=1$ ($\xi=1,\dots,z$) if the embedding selected for $skel(\mu_\xi)$ is its second embedding. Hence, this part of the tuple specifies the embeddings of all the skeletons of the R-nodes of~$\mathcal{T}^T$.
Consider now the elements $p_1, \dots, p_y$ of the tuple in the statement. Each of them is associated with a P-node of~$\mathcal{T}^T$. We exploit a result \cite{DBLP:journals/ipl/MyrvoldR01} of Myrvold and Ruskey to put each natural number $p_\xi$ ($\xi=1,\dots,y$) in one-to-one correspondence, in linear time in $\delta(p_\xi)$, to a permutation of the edges of $skel(\nu_\xi)$, associating the value $0$ with its first embedding. Notice that the circular order of the edges of $skel(\nu_\xi)$ corresponds to a permutation of all its edges with the exception of the reference edge, thus we have $(\delta(\nu_\xi)-1)!$ possible permutations.

Given $G$, tree~$\mathcal{T}^T$ can be computed in $O(n)$ time \cite{DBLP:journals/siamcomp/BattistaT96,DBLP:conf/gd/GutwengerM00}.
Given a planar embedding $\cal E$ of $G$ the corresponding tuple $\langle p_1, \dots, p_y, r_1, \dots, r_z \rangle$ can be computed in $O(n\log n)$ time: (1) by assigning to each $r_\xi$ a value depending on the fact that the corresponding skeleton, related to $\cal E$, has its first or second embedding and (2) by assigning to each $p_\xi$ a natural number corresponding to its permutation in $\cal E$ using \cite{DBLP:journals/ipl/MyrvoldR01}. Observe that the number of edges in the skeletons of the P-nodes in $\mathcal{T}^T$ is linear and therefore the number stored in $p_\xi$ may need $O(n \log n)$ bits to be represented.
Conversely, given a tuple $\langle p_1, \dots, p_y, r_1, \dots, r_z \rangle$ one can build an embedding for $G$ by just selecting one of the two embeddings of the skeletons of R-nodes according to the values of the $r_\xi$ and by converting, for each $p_\xi$ the given natural number into an embedding of the skeleton of the corresponding P-node.
\end{proof}

We have focused on a specific biconnected component, but the same approach can be applied to all the biconnected components of a graph. Also, the embeddings of the R-nodes and P-nodes in the SPQR-trees of these components are independent of each other.

\section{Connected Graphs}\label{se-app:connected}
Here we give the proof of the theorems stated in \cref{se:simply-connected}
\thConnected*
\begin{proof}
First, from \cite{cai1993counting} we know that all the planar embeddings on the sphere of the $b(v)$ biconnected components around $v$ are $E_v = \prod_{j=1}^{b(v)} \delta_{v,j} \prod_{j=1}^{b(v)-2}(\delta_v-j)$, and this number corresponds to all the possible tuples $\langle c_{1}, \dots, c_{b(v)}, d_{1}, \dots, d_{b(v)-2} \rangle$. Having those two sets the same cardinality, there is a bijection between them. We need to prove that this bijection is $\varphi^{-1}_v$ as described in \cref{ssec:tuple2embedding}. To do so it is sufficient to show that $\varphi^{-1}_v$ is injective, as having the domain and codomain the same cardinality, this implies surjectivity. Then we need to demonstrate that $\varphi_v$, as described in \cref{ssec:embedding2tuple}, is the inverse function of $\varphi^{-1}_v$.
Finally we will discuss the time needed to compute the two functions.

Given two different tuples $\tau=\langle c_1,\dots,c_{b(v)},d_1,\dots,d_{b(v)-2}\rangle$ and $\tau'=\langle c_1',$ $\dots,c_{b(v)}',d_1',$ $\dots,d_{b(v)-2}'\rangle$, suppose that $\varphi_v^{-1}(\tau)=\mathcal{E}=\varphi_v^{-1}(\tau')=\mathcal{E'}$, i.e. they correspond to the same embedding.
Let $k\in [2\dots b(v)]$ be the first index such that $c_k\neq c_k'$. In that case, in a counter-clockwise traversal of the incidence list of $v$ starting from a random edge of $B_1$, the first edge of block $B_k$ encountered will be $c_k$ in $\mathcal{E}$ and $c_k$ in $\mathcal{E}'$, but as $c_k\neq c_k'$ this means that $\mathcal{E}\neq\mathcal{E}'$, a contradiction. In the case in which $k=1$ it is sufficient to start the traversal from a random edge of $B_2$.
Now, assuming that all the $c$ values are the same in $\tau$ and $\tau'$, let $k\in [1\dots b(v)-2]$ be the first index such that $d_k\neq d_k'$. We distinguish three cases depending on $d_k$ and $d_k'$ referring to edges that belong to $B_k$ or to some component that has attached to it.
Let $e_{d_k}\in B_i$ and $e_{d_k}'\in B_j$ such that $i,j\leq k$ and such that their are both considered to belong to $\mathcal{P}^\ssquare_k$. As $e_{d_k}$ (resp. $e_{d_k}'$) is the first edge of $B_k$ to be encountered after all the edges of $B_2$ in a counter-clockwise traversal of the incidence list of $v$ starting from a random edge of $B_1$ in $\mathcal{E}$ (resp. $\mathcal{E'}$), and $e_{d_k}\neq e_{d_k}'$, as they can not be first edges of some component, we have the contradiction $\mathcal{E}\neq\mathcal{E}'$.
Let now $e_{d_k}\in B_i$ with $i\leq k$ and $e_{d_k}$ considered to belong to $\mathcal{P}^\ssquare_k$, while $e_{d_k}'\in B_j$ with $j>k$ belongs to $\mathcal{P}^\ssquare_j$. In this case, $first_k$ is put right after some edge $e^*$ that precedes $first_2$ in $\mathcal{P}^\ssquare_1$, while it is put right after $e_{d_k}'$ in $\mathcal{P}^\ssquare_k$. As $e^*$ must belong to a component $B_h$ whose partial embedding has been merged with $\mathcal{P}^\ssquare_1$, we have that $h<i$ and consequently $e^*\neq e_{d_k}'$ and again $\mathcal{E}\neq\mathcal{E}'$.
Let finally $e_{d_k}\in \mathcal{P}^\ssquare_i$ and $e_{d_k}'\in  \mathcal{P}^\ssquare_j$ such that $i,j \neq k$. As $first_h$ for each $2<h<k$ has been substituted in $\mathcal{S}$ to the edges selected using $d_h$, $e_{d_k}\neq e_{d_k}'$ and as no other edges will be put between $first_k$ and $e_{d_k}$ or $e_{d_k}'$ we have the contradiction $\mathcal{E}\neq\mathcal{E}'$.
This concludes the demonstration that $\varphi^{-1}_v$ is injective and consequently it constitutes a bijection. 

We now prove that $\varphi_v$ is the inverse of $\varphi_v^{-1}$. Consider a tuple $\tau$ and an embedding $\mathcal{E}=\varphi_v^{-1}(\tau)$, we need to demonstrate that $\tau'=\varphi_v(\mathcal{E})=\tau$.
Each number $c_i\in \tau$, for $i=1,\dots,b(v)$, is used by $\varphi_v^{-1}$ to compute a planar embedding on the plane of each biconnected component incident to $v$, defining an edge $first_i$ for each of them. During the merge operations those embeddings are preserved, and after $b(V)-1$ merging operations the result is an embedding on the plane $\mathcal{P}^\ssquare_1$ that can be transformed again into the embedding on the sphere $\mathcal{E}$. In $\mathcal{E}$ all the edges $first_i$ for $i=2,\dots,b(v)$ can be computed by traversing counter-clockwise the edges incident to $v$ starting from $first_1$. To find $first_1$ we can exploit the fact that in the embedding on the plane $\mathcal{P}^\ssquare_1$ the embeddings $\mathcal{E}^\ssquare_1$ and $\mathcal{E}^\ssquare_2$ have been merged in such a way that their edges are not interleaved and $first_2$ comes after $last_1$. In fact, starting a counter-clockwise visit of the edges around $v$ from a random edge of $B_2$, the first edge of $B_1$ that will be found is $first_1$. As this procedure is used by $\varphi_v(\mathcal{E})$ to compute the elements $c_1',\dots, c_{b(v)}'$ of $\tau'$, we can conclude that $c_1'=c_1,\dots, c_{b(v)}'=c_{b(v)}$.

Consider now a partial embedding $\mathcal{P}^\ssquare_i$ during the computation of $\varphi_v^{-1}(\tau)$. If it falls into case 2, i.e. the edge $e_{d_i}$ is an edge that belongs to $\mathcal{P}^\ssquare_i$, it merges with $\mathcal{P}^\ssquare_1$ by inserting $first_i$ right before $first_2$ (after the edge $e^*$) and the edge $e_{d_i}$ right after $last_2$. If there is a partial embedding $\mathcal{P}^\ssquare_j$ with $j>i$ that also falls into case 2, this will be merged with $\mathcal{P}^\ssquare_1$ in the same way, resulting in being nested between $\mathcal{P}^\ssquare_i$ and $\mathcal{P}^\ssquare_1$ as it was before the merging with $\mathcal{P}^\ssquare_i$. Due to the change of the edge $e^*$ before $first_2$ with $first_i$ in $\mathcal{S}$ after the merge, no partial embedding $\mathcal{P}^\ssquare_k$ with $k>i$ could be put completely between $e^*$ and $first_i$. During the computation of $\varphi_v(\mathcal{E})$, consider the node $n(B_i)$ in the tree $\mathcal{T}^B$. It belongs to the subtree rooted at a child of $\rho$ that comprehends $n(B_2)$. Further it either belongs to the path $\pi_2$ (as originally defined), having part of its child on the left of the path and part of the child on the right of the path, either it is totally on the left of the path, but its right sibling is a node $n(B_j)$ with $j<i$ such that either $n(B_j)\in \pi_2$, either it has another right sibling $n(B_k)$ with $k<j$ with the same properties, eventually reaching a node $n(B_h)$ on $\pi_2$. Having that $h<i$, and being $first_h$ between $first_i$ and $first_2$, we have necessarily that $\mathcal{P}^\ssquare_h$ has been merged with $\mathcal{P}^\ssquare_i$, falling into case 1 during the computation of $\varphi_v^{-1}(\tau)$ and so for all the right siblings of $n(B_i)$ in $\mathcal{T}^B$ between $n(B_i)$ and $n(B_h)$. On the other side, this also implies that $d_i$ corresponds to an edge that is part of $\mathcal{E}^\ssquare_h$ and after the merges is in $\mathcal{P}^\ssquare_i$. This edge $e$ is the one corresponding to the first edge-node child of $n(B_h)$ on the right of $\pi_2$ in $\mathcal{T}^B$. Similarly, if $n(B_i)$ is on $\pi_2$, all the nodes in $\pi_2$ that are descendants of $n(B_i)$ and not descendants of a node $n(B_j)$, with $j>i$, that is also a descendant of $n(B_i)$ in $\pi_2$, correspond to partial embeddings that have been merged with $\mathcal{P}^\ssquare_i$, thus falling into case 1. In this case, $d_i$ pointed to the first edge $e$ on the right of $last_2$ of the deepest of those nodes in $\mathcal{T}^B$. Observe that during the computation of $\varphi_v(\mathcal{E})$ the label $jump_i$ corresponds to the original label of the edge $e$, as the labels of the edge-nodes have not been changed yet. Furthermore, the reverse post-order traversal used to update the label $jump_i$ follows exactly those rules, firstly considering the existence of a left sibling and then the existence of an ancestor with indexes greater than $i$. In fact, if both cases apply, the existence of the ancestor applies also to the left sibling in the following step of the traversal. Thus the numbers $d_1,\dots,d_{b(v)-2}$ such that during the computation of $\varphi_v^{-1}(\tau)$ fall into case 2 are the same as the values $d_1',\dots,d_{b(v)-2}'$ computed by the application of case 2 in $\varphi_v^(\mathcal{E})$.
If instead $\mathcal{P}^\ssquare_i$ falls into case 1 during the computation of $\varphi_v^{-1}(\tau)$, $d_i$ may correspond either to an edge that has not been used for merging operations or to an edge $first_j$ for some $j<i$ such that $first_j$ has substituted (possibly indirectly) the edge that was originally in position $d_i$ in $\mathcal{S}$ after a merging operation. When merging $\mathcal{P}^\ssquare_i$ with the partial embedding $\mathcal{P}^\ssquare_j$ to which the edge $e$ referenced by $d_i$ belongs, $\mathcal{P}^\ssquare_i$ is put completely between $e$ and the edge (possibly null) that follows $e$ in $\mathcal{P}^\ssquare_j$, and in particular is such that nothing else will be put between $e$ and $first_i$. Let $B_j$ be the block containing $e$, we have that, during the computation of $\varphi_v^(\mathcal{E})$, the node $n(B_i)$ is either child of $n(B_j)$ or a right sibling of $n(B_j)$ (in the case in which $e$ is $last_j$). If $e$ is not $first_j$, we have that its label $\ell(e)$ has not been changed and corresponds to the position of $e$ in $\mathcal{S}$ and this is exactly $d_i$ as the list $\mathcal{S}$ is computed in the same way in $\varphi_v$ and $\varphi_v^{-1}$, so $d_i'=d_i$. Otherwise, if $e$ is $first_j$, we may have that the label $\ell(e)$ has been modified by previous computation steps. If this has not been modified, i.e. $j<3 \lor j>i$, we have the previous situation and then $d_i'=d_i$. If $\ell(e)$ has been modified, there exists a component $B_k$ with $k<j$ such that $n(B_k)$ is either the parent or the right sibling of $n(B_j)$ and an edge $e_k$ that belongs to $B_k$ and comes right before $first_j$ in $\mathcal{E}$. Again if $e_k$ is $first_k$ we have the same situation, and eventually we reach a node-component $n(B_h)$ such that $h<3 \lor h>i$. The sequence of changes makes $\ell(e)=\ell(e_h)$. This corresponds to the sequence of merging that involved the same number $d_i$ in $\varphi_v^{-1}(\tau)$ more than one time. In fact, having embedded $\mathcal{E}^\ssquare_i$ right after $first_j$ with $j<i$ implies that the number $d_i$ was used also for the merging of $\mathcal{P}^\ssquare_j$ and thereafter $first_j$ was put in position $d_i$ in $\mathcal{S}$. Having changed the label $\ell(first_j)=\ell(e_k)=\ell(e_h)$ after the computation of $d_j$, we have that $d_i'=\ell(e_h)=d_i$.
This concludes the proof of $\varphi_v$ and $\varphi_v^{-1}$ being one the inverse function of the other.

Regarding the computation times, notice first that all the numbers involved have a linear size with respect to the number of edges in the graph, therefore we can avoid considerations on the space and time needed for their representation.
In $\varphi_v$ the computation of the values $c_i$ and $first_i$ takes linear time, as it can be performed by visiting twice the edges incident to $v$. Similarly, the labeling of the edges and successive creation of $\mathcal{S}$ can be performed by visiting twice all the edges incident to $v$ for each block. The creation of $\mathcal{T}^B$ also takes linear time, as we perform for each edge incident to $v$ an operation that takes constant time, assuming that an edge knows to which component it belongs and the component has references to its first and last edges. The computation of $\pi_2$ and the pointers $jump_i$ takes linear time, as they can be linearly many and they can be created by a top-down traversal of $\pi_2$. The reverse post-order traversal also takes linear time as the operations to be performed on each node take constant time. Finally, the computation of the values $d_i$ takes constant time for each of them, as the condition for deciding between case 1 and case 2 can be checked in constant time, and in both cases the operations performed are constant. As all the operations described above are performed in sequence, we have that computing $\varphi_v$ takes $O(n)$ time.

In $\varphi_v^{-1}$ the computation of the embeddings on the plane at the beginning of the algorithm requires $O(n)$ time, as we must order the edges incident to $v$ for each component by their identifiers and a component may comprehend linearly many edges. To this aim, we could use a bucket sort or similar linear time algorithms. The labeling that is computed thereafter so that in each component $B_i$ the edges are ordered starting from $first_i$ takes linear time, as we can traverse the edges incident to $v$ in $B_i$ starting from $first_i$. This allows us to compute also $last_i$ and $\mathcal{S}$. The merging operation can be performed in constant time, by just merging the incidence lists on $v$ in the two partial embeddings, according to what is described in the two cases. Also, the update on $\mathcal{S}$ can be performed in constant time. The only remaining part is the check of whether an edge belongs to a partial embedding or not. This can be done in total $O(n\alpha(n))$, where $\alpha$ is the inverse Ackermann function, using a union-find data structure as described in \cite{10.1145/62.2160}. Considering that all those operations are independently performed, we have a total cost of $O(n\alpha(n))$.
\end{proof}

Generalizing to multiple cut-vertices, we have the following.

\leConnected*
\begin{proof}
    We produce $\varphi$ by showing how to obtain the sequence $\omega_1, \omega_2, \dots, \omega_w$, where $\omega_i \in [0 \dots E_{v_i}-1]$, from an embedding on the sphere of $G$ and vice-versa.
    
    Let $\cal E$ be an embedding on the sphere of $G$ such that the restriction of $\cal E$ to each block $B_j$ is $\mathcal{E}_j$, $j=1, \dots, x$. Consider a cut-vertex $v_i$ of $G$ and remove from $\cal E$ the biconnected components that are not incident to $v_i$. Apply the function $\varphi_{v_i}$ to the obtained arrangement around $v_i$ of the biconnected components incident to $v_i$ and you have a number $\omega_i \in [0 \dots  E_{v_i}-1]$. Iterating for each $i=1, \dots, w$ you have a sequence of numbers $\omega_1, \omega_2, \dots, \omega_w$ where $\omega_i \in [0 \dots E_{v_i}-1]$. 
    
    Conversely, let $\omega_1, \omega_2, \dots, \omega_h$ be a sequence of numbers such that $\omega_j \in [0 \dots E_{v_j}-1]$ for $j=1, \dots, h$. We construct $\cal E$ incrementally. Initialize $\cal E$ by constructing with $\varphi_{v_1}^{-1}(\omega_1)$ the arrangement around $v_1$ of the embeddings of the biconnected components incident to $v_1$. Let $v_\ell$ be a cut vertex of $G$ belonging to some biconnected components of the current $\cal E$. Observe that $v_\ell$ belongs to a single biconnected component $B_p$ of the current $\cal E$, otherwise, if $v_\ell$ belonged to two biconnected components $B_p$ and $B_q$ of $\cal E$, there would be a cycle traversing $B_p$ and $B_q$. We use $\varphi_{v_\ell}^{-1}(\omega_\ell)$ to construct the arrangement of the biconnected components incident to $v_\ell$, where ${\cal E}_p$ is replaced by the current $\cal E$. The produced arrangement is used to update $\cal E$ and the procedure is iterated until there is some non-processed cut-vertex of $G$ in the current $\cal E$. 
\end{proof}

\section{Examples of ranking and unranking blocks around a cut-vertex}
\begin{figure}[h!]
    \captionsetup[subfigure]{justification=centering}
    \centering
    \begin{subfigure}{0.55\textwidth}
    \centering
    \includegraphics[page=1, width=\textwidth]{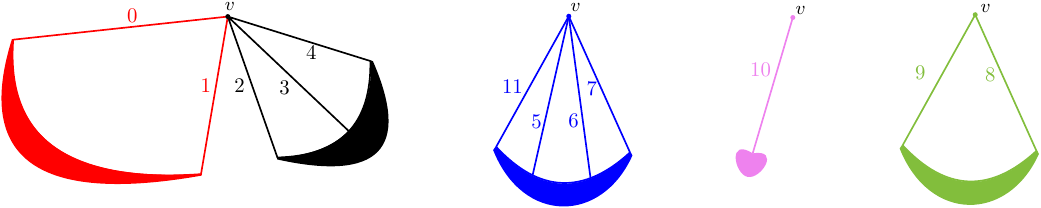}
    \subcaption{The components after the unranking of the first edges.}
    \end{subfigure}
    \hfil
    \begin{subfigure}{0.4\textwidth}
    \centering
    \includegraphics[page=2, width=\textwidth]{figures/ranking-cut-vertex.pdf}
    \subcaption{$\mu_3$ has attachment number 1 and so its first edge ($e_{12}$) unifies with $e_1$.}
    \end{subfigure}
    \hfil
    \begin{subfigure}{0.45\textwidth}
    \centering
    \includegraphics[page=3, width=\textwidth]{figures/ranking-cut-vertex.pdf}
    \subcaption{The only edge of $\mu_4$ ($e_{11}$) unifies with $e_1$, after $e_{12}$, so the component will lie inside the face bounded by $e_{12}$ and $e_6$.}
    \end{subfigure}
    \hfil
    \begin{subfigure}{0.35\textwidth}
    \centering
    \includegraphics[page=4, width=\textwidth]{figures/ranking-cut-vertex.pdf}
    \subcaption{The first edge of $\mu_5$ ($e_{10}$) unifies with $e_1$ after $e_{11}$.}
    \end{subfigure}
    \caption{The unranking of the embedding 0-0-0.}\label{fig:unranking-0-0-0}
\end{figure}

\begin{figure}[h!]
    \captionsetup[subfigure]{justification=centering}
    \centering
    \begin{subfigure}{0.55\textwidth}
    \centering
    \includegraphics[page=6, width=\textwidth]{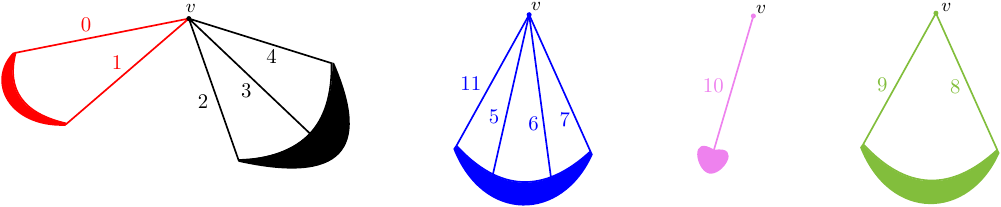}
    \subcaption{The components after the unranking of the first edges.}
    \end{subfigure}
    \hfil
    \begin{subfigure}{0.4\textwidth}
    \centering
    \includegraphics[page=7, width=\textwidth]{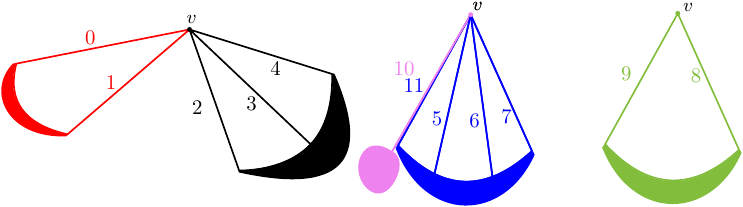}
    \subcaption{$\mu_3$ has attachment number 11 and so its first edge ($e_{12}$) unifies with $e_{11}$.}
    \end{subfigure}
    \hfil
    \begin{subfigure}{0.45\textwidth}
    \centering
    \includegraphics[page=8, width=\textwidth]{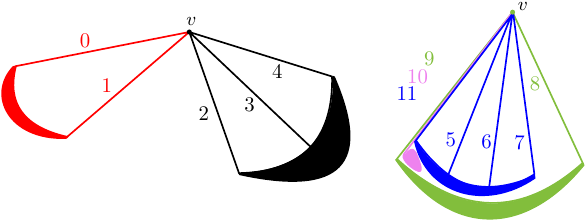}
    \subcaption{The only edge of $\mu_4$ ($e_{11}$) unifies with $e_{10}$, carrying $\mu_3$ with it.}
    \end{subfigure}
    \hfil
    \begin{subfigure}{0.35\textwidth}
    \centering
    \includegraphics[page=9, width=\textwidth]{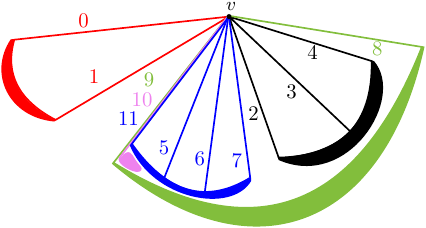}
    \subcaption{The first edge of $\mu_5$ ($e_{10}$) unifies with the last edge before $\mu_2$ and the edge $e_9$ goes right after $e_5$.}
    \end{subfigure}
    \caption{The unranking of the embedding 10-9-8.}\label{fig:unranking-10-9-8}
\end{figure}

\begin{figure}[!h]
    \centering
    \includegraphics[page=5, width=\textwidth]{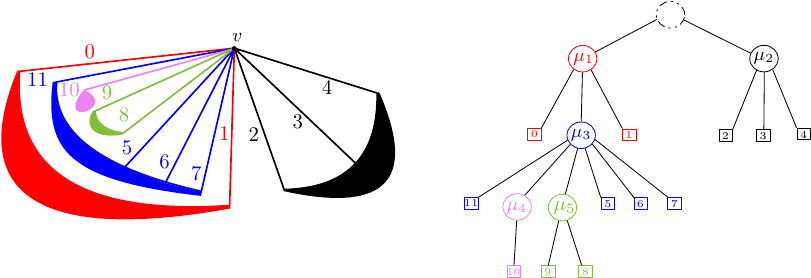}
    \caption{The planar embedding on a sphere around the cut-vertex $v$ and its resulting nesting-tree. The ranking will compute the numbers 0-0-0.}
    \label{fig:ranking-0-0-0}
\end{figure}

\begin{figure}[!h]
    \centering
    \includegraphics[page=10, width=\textwidth]{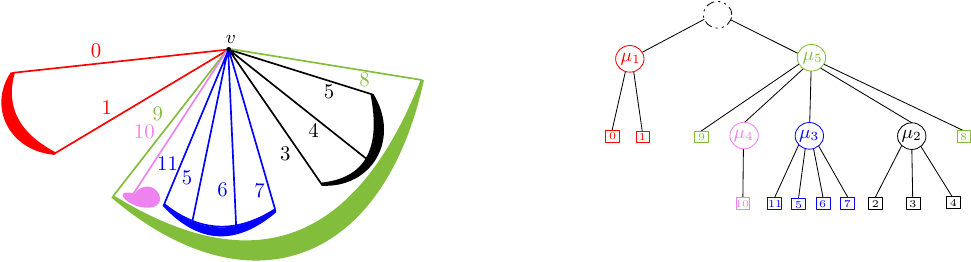}
    \caption{The planar embedding on a sphere around the cut-vertex $v$ and its resulting nesting-tree. The ranking will compute the numbers 10-9-8.}
    \label{fig:ranking-10-9-8}
\end{figure}

\clearpage

\section{Examples of ranking and unranking connected components}

\begin{figure}[h!]
    \captionsetup[subfigure]{justification=centering}
    \centering
    \begin{subfigure}{0.15\textwidth}
    \centering
    \includegraphics[page=1, width=\textwidth]{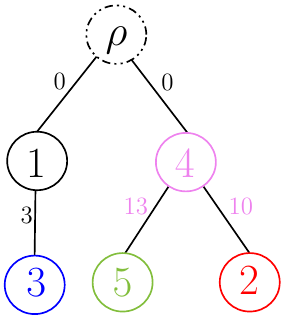}
    \subcaption{$\varnothing$}
    \end{subfigure}
    \hfil
    \begin{subfigure}{0.15\textwidth}
    \centering
    \includegraphics[page=2, width=\textwidth]{figures/pruferone.pdf}
    \subcaption{$10$}
    \end{subfigure}
    \hfil
    \begin{subfigure}{0.15\textwidth}
    \centering
    \includegraphics[page=3, width=\textwidth]{figures/pruferone.pdf}
    \subcaption{$10, 3$}
    \end{subfigure}
    \hfil
    \vspace{1pc}
    \begin{subfigure}{0.15\textwidth}
    \centering
    \includegraphics[page=4, width=\textwidth]{figures/pruferone.pdf}
    \subcaption{$10, 3, 0$}
    \end{subfigure}
    \hfil
    \begin{subfigure}{0.15\textwidth}
    \centering
    \includegraphics[page=5, width=\textwidth]{figures/pruferone.pdf}
    \subcaption{$10, 3, 0, 13$}
    \end{subfigure}
    \hfil
    \begin{subfigure}{0.15\textwidth}
    \centering
    \includegraphics[page=6, width=\textwidth]{figures/pruferone.pdf}
    \subcaption{$10, 3, 0, 13$}
    \end{subfigure}
    \caption{The ranking of a nesting tree using the variation of the Pr{\"u}fer sequence. The algorithm produces the tuple $\tau=\langle 10, 3, 0, 13 \rangle$}\label{fig:nest-tree-to-pruferone}
\end{figure}

\begin{figure}[h!]
    \captionsetup[subfigure]{justification=centering}
    \centering
    \begin{subfigure}{0.25\textwidth}
    \centering
    \includegraphics[page=7, width=\textwidth]{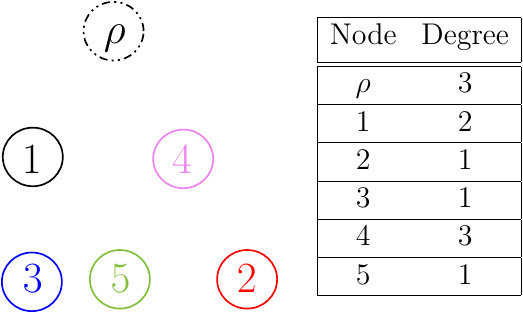}
    \end{subfigure}
    \hfil
    \begin{subfigure}{0.25\textwidth}
    \centering
    \includegraphics[page=8, width=\textwidth]{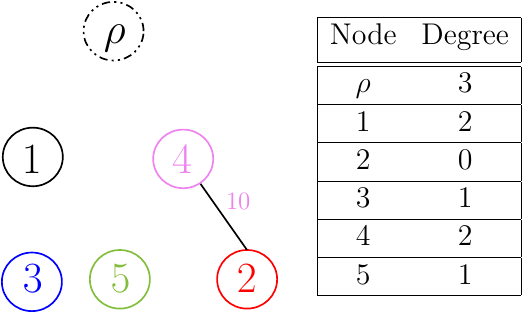}
    \end{subfigure}
    \hfil
    \begin{subfigure}{0.25\textwidth}
    \centering
    \includegraphics[page=9, width=\textwidth]{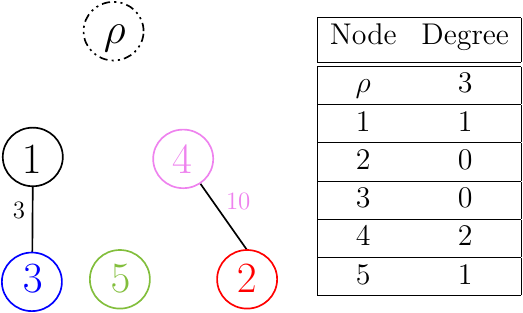}
    \end{subfigure}
    \hfil
    \vspace{1pc}
    \begin{subfigure}{0.25\textwidth}
    \centering
    \includegraphics[page=10, width=\textwidth]{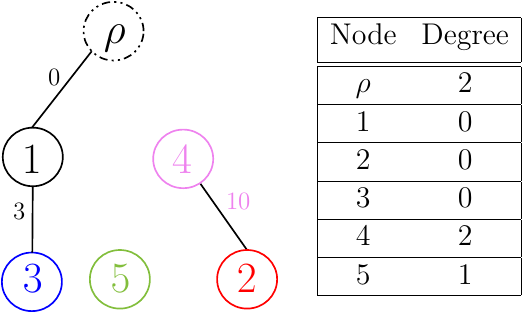}
    \end{subfigure}
    \hfil
    \begin{subfigure}{0.25\textwidth}
    \centering
    \includegraphics[page=11, width=\textwidth]{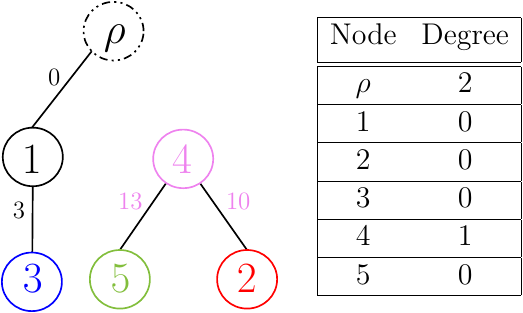}
    \end{subfigure}
    \hfil
    \begin{subfigure}{0.25\textwidth}
    \centering
    \includegraphics[page=12, width=\textwidth]{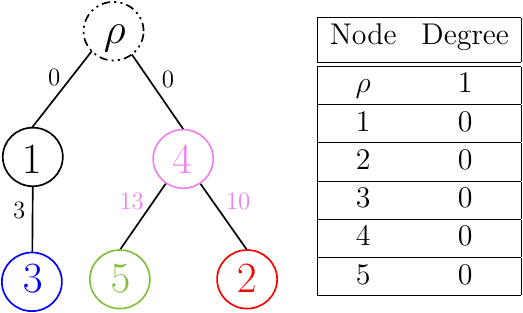}
    \end{subfigure}
    \caption{The unranking of a nesting tree using the tuple $\tau=\langle 10, 3, 0, 13 \rangle$. The tables representing the association between a component $G_i$ and $\delta_i$. The tuple $\tau'$ is equal to $\langle 4, 1, 0, 4 \rangle$. The labels of the faces of $G_i$ are as follows: $G_1 = \{1,2,3\}, G_2 = \{4,5\}, G_3 = \{6,7,8,9\}, G_4 = \{10, 11, 12, 13\}, G_5 = \{14, 15\}$.} \label{fig:pruferone-to-nest-tree}
\end{figure}

\begin{figure}[h!]
    \centering
    \includegraphics[page=13, width=0.5\textwidth]{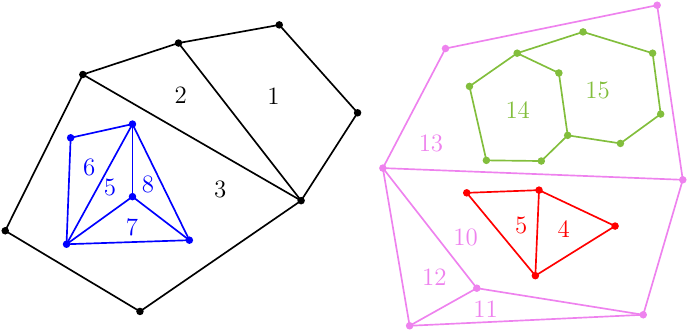}
    \caption{An example of a non-connected graph whose embedding corresponds to the nesting tree in \cref{fig:nest-tree-to-pruferone,fig:pruferone-to-nest-tree}, the components and the labeling of their internal faces are highlighted by the colors}
    \label{fig:non-connected-example}
\end{figure}

%
%
%
%
\end{document}